\documentclass[a4paper,UKenglish]{lipics}
 
\usepackage{microtype} 
\usepackage[ruled,vlined]{algorithm2e}
\usepackage{cite}
\usepackage{wrapfig}

\newcommand{\stitle}[1]{\vspace{0.5em}\noindent\textbf{#1}}
\newcommand{\squishlist}{
   \begin{list}{$\bullet$}
    { \setlength{\itemsep}{0pt}
      \setlength{\parsep}{2pt}
      \setlength{\topsep}{2pt}
      \setlength{\partopsep}{0pt}
    }
}
\newcommand{\squishend}{\end{list}}
\newcommand{\mA}{\mathcal{A}}
\newcommand{\mR}{\mathcal{R}}
\newcommand{\mC}{\mathcal{C}}
\newcommand{\mB}{\mathcal{B}}
\newcommand{\mO}{\mathcal{O}}
\newcommand{\mH}{\mathcal{H}}
\newcommand{\mV}{\mathcal{V}}
\newcommand{\mE}{\mathcal{E}}
\newcommand{\mT}{\mathcal{T}}
\newcommand{\mZ}{\mathcal{Z}}

\newcommand{\attr}{\mathsf{attr}}
\newcommand{\IN}{\mathrm{IN}}
\newcommand{\OUT}{\mathrm{OUT}}
\newcommand{\AGM}{\mathsf{AGM}}
\newcommand{\DBP}{\mathsf{DBP}}
\newcommand{\degree}{\mathsf{deg}}
\newcommand{\MW}{\mathsf{MW}}
\newcommand{\MO}{\mathsf{MO}}

\theoremstyle{definition}
\newtheorem{proposition}[theorem]{Proposition}
\newtheorem*{theorem*}{Theorem}
\newtheorem*{lemma*}{Lemma}
\newtheorem*{proposition*}{Proposition}
\newcounter{prob}
\newtheorem{problem}[prob]{Problem}
\newtheorem{lp}{Linear Program}

\volumeinfo{Wim Martens and Thomas Zeume}
{2}
{19th International Conference on Database Theory (ICDT 2016)}
{48}
{1}
{1}
\EventShortName{ICDT'16}%
\DOI{10.4230/LIPIcs.ICDT.2016.1}
\serieslogo{}
\title{It's all a matter of degree: Using degree information to optimize multiway joins}

\author[1]{Manas R. Joglekar}
\author[2]{Christopher M. R\'e}
\affil[1]{Department of Computer Science, Stanford University\\
  450 Serra Mall, Stanford, CA, USA\\
  \texttt{\{manasrj, chrismre\}@stanford.edu}}
\authorrunning{M.\,R. Joglekar and C.\,M. R\'e} 

\Copyright{Manas R. Joglekar and Christopher M. R\'e}


\serieslogo{}
\volumeinfo
  {Billy Editor and Bill Editors}
  {2}
  {Conference title on which this volume is based on}
  {1}
  {1}
  {1}
\EventShortName{}
\DOI{10.4230/LIPIcs.xxx.yyy.p}

\begin{document}

\maketitle

\begin{abstract}
We optimize multiway equijoins on relational tables using degree
information. We give a new bound that uses degree information to more
tightly bound the maximum output size of a query. On real data, our
bound on the number of triangles in a social network can be up to $95$
times tighter than existing worst case bounds. We show that using only
a constant amount of degree information, we are able to obtain join
algorithms with a running time that has a smaller exponent than
existing algorithms--{\em for any database instance}. We also show
that this degree information can be obtained in nearly linear time,
which yields asymptotically faster algorithms in the serial setting
and lower communication algorithms in the MapReduce setting.

In the serial setting, the data complexity of join processing can be
expressed as a function $O(\IN^x + \OUT)$ in terms of
input size $\IN$ and output size $\OUT$ in which $x$
depends on the query. An upper bound for $x$ is given by fractional
hypertreewidth. We are interested in situations in which we can get
algorithms for which $x$ is strictly smaller than the fractional
hypertreewidth. We say that a join can be processed in subquadratic
time if $x < 2$. Building on the AYZ algorithm for processing cycle
joins in quadratic time, for a restricted class of joins which we call
$1$-series-parallel graphs, we obtain a complete decision procedure
for identifying subquadratic solvability (subject to the $3$-SUM
problem requiring quadratic time). Our $3$-SUM based quadratic lower
bound is tight, making it the only known tight bound for joins that
does not require any assumption about the matrix multiplication
exponent $\omega$. We also give a MapReduce algorithm that meets our
improved communication bound and handles essentially optimal
parallelism.
\end{abstract}

\section{Introduction}
We study query evaluation for natural join queries. Traditional
database systems process joins in a pairwise fashion (two tables
at a time), but recently a new breed of multiway join algorithms have
been developed that satisfy stronger runtime guarantees. In the
sequential setting, worst-case-optimal sequential algorithms such as
NPRR~\cite{Ngo:2012:WOJ:2213556.2213565,ngo:survey}
or LFTJ~\cite{DBLP:journals/corr/abs-1210-0481} process the join in runtime that is upper bounded
by the largest possible output size, a stronger guarantee than what
traditional optimizers provide. In MapReduce settings (described in Appendix~\ref{sec:mapreduce}), the Shares algorithm
~\cite{afrati:mrjoins,koutris} (described in Appendix~\ref{sec:shares}) processes multiway joins with optimal communication complexity on skew free data. However, traditional
database systems have developed sophisticated techniques to improve
query performance. One popular technique used by commercial database
systems is to collect ``statistics'': auxiliary information about
data, such as relation sizes, histograms, and counts of distinct different
attribute values. Using this information helps the system
better estimate the size of a join's output and the runtimes of different query plans, and make better choices
of plans. Motivated by the use of statistics in query processing, we
consider how statistics can improve the new breed of multiway join
algorithms in sequential and parallel settings. 

We consider the first natural choice for such statistics about the
data: the degree. The degree of a value in a table is the number of
rows in which that value occurs in that table. We describe a simple preprocessing technique to facilitate 
the use of degree information, and demonstrate its value through three applications: 
i) An improved output size bound ii) An improved sequential join algorithm iii) An improved MapReduce join algorithm. 
Each of these applications has an improved exponent relative to their corresponding state-of-the-art versions~\cite{atserias:agm,Ngo:2012:WOJ:2213556.2213565,DBLP:journals/corr/abs-1210-0481,beame:skew}.

Our key technique is what we call {\em degree-uniformization}. Assume
for the moment that we know the degree of each value in each relation, we
then partition each relation by degree of each of its attributes. In particular, we
assign each degree to a bucket using a parameter
$L$: we create one bucket for degrees in $[1, L)$, one for degrees in
$[L, L^2)$, and so on. We then place each tuple in every relation into a partition based on
the degree buckets for each of its attribute values. The join problem then naturally splits into 
smaller join problems; each smaller problem consisting of a join using one partition from each
relation. Let $\IN$ denote the input
size, if we set $L=\IN^{c}$ for some constant $c$, say $\frac{1}{4}$, the number of
smaller joins we process will be exponential in the number of relations--but
constant with respect to the data size $\IN$. Intuitively, the benefit of
joining partitions separately is that each partition will have more
information about the input and will have reduced skew. We show
that by setting $L$ appropriately this scheme allows us to get tighter
AGM-like bounds. 

Now we consider a concrete example. Suppose we have a $d$-regular graph with $N$ edges; the number of triangles in the graph is bounded by $\min(Nd, \frac{N^2}{d})$ by our degree-based bound and by $N^{3/2}$ by the AGM bound. In the worst case, $d = \sqrt{N}$ and our bound matches the AGM bound. But for other degrees, we do much better; 
better even than simply ``summing'' the AGM bounds over each combination of partitions. 
Table~\ref{table:dbp-vs-agm} compares our bound ($\MO$) with the AGM bound for the triangle
join on social networks from the SNAP datasets~\cite{snapnets}. `M' in the table stands for millions. The last column shows the ratio of the AGM bound to our bound; our bound is tighter by a factor of $11x$ to $95x$.
We could not compare the bounds on the Facebook network, but if the number of friends per user is $\leq 5000$, our bound is
at least $450x$ tighter than the AGM bound.

\begin{wraptable}{l}{83mm}
    \begin{tabular}{| l | c | c | c |}
      \hline
      {\bf Network} & {\bf $\MO$ Bound} & {\bf $\AGM$ Bound} & $\frac{\AGM}{\MO}$ \\ \hline
      Twitter & $225M$ & $3764M$ & $17$ \\ \hline
      Epinions & $33M$ & $362M$ & $11$ \\ \hline
      LiveJournal & $6128M$ & $573062M$ & $95$ \\ \hline
    \end{tabular}
\caption{Triangle bounds on various social networks} \label{table:dbp-vs-agm}
\vspace{-20pt}
\end{wraptable}

We further use degree uniformization as a tool to develop
algorithms that satisfy stronger runtime and communication guarantees. 
Degree uniformization allows us to get runtimes with a better exponent than existing algorithms,
while requiring only linear time preprocessing on the data. 
We demonstrate our idea in both the serial and parallel (MapReduce)
setting, and we now describe each in turn.

\stitle{Serial Join Algorithms:}
We use our degree-uniformization to derive new cases in which one can
obtain subquadratic algorithms for join processing. More precisely,
let $\IN$ denote the size of the input, and $\OUT$
denote the size of the output. Then the runtime of an algorithm on a query
$Q$ can be written as $O(\IN^{x} + \OUT)$ for some $x$. 
Note that $x \geq 1$ for all algorithms and queries in this model as we must read
the input to answer the query. If the query is $\alpha$-acyclic, Yannakakis'
algorithm~\cite{Yannakakis81} achieves $x=1$. If the query has 
fractional hypertree width (fhw), a recent generalization of tree
width~\cite{gottlob:ght}, equal to $2$, then we can achieve $x = 2$ using a
combination of algorithms like NPRR and LFTJ with Yannakakis'
algorithm. In this work, we focus on cases for
which $x < 2$, which we call {\it subquadratic
  algorithms}. Subquadratic algorithms are interesting creatures in
their own right, but they may provide tools to attack the common case
in join processing in which $\OUT$ is smaller than
$\IN$.

Our work builds on the classical AYZ
algorithm~\cite{Alon:1994:FCG:647904.739463}, which derives
subquadratic algorithms for cycles using degree information. This is a
better result than the one achieved by the fhw result since the fhw
value of length $\geq 4$ cycles is already $= 2$. This result is
specific to cycles, raising the question: {\it ``Which joins are
  solvable in subquadratic time?''} Technically, the AYZ algorithm
makes use of properties of cycles in their result and of
``heavy and light'' nodes (high degree and low degree,
respectively). We show that degree-uniformization is a generalization
of this method, and  that it allows us to derive subquadratic algorithms for a
larger family of joins. We devise a procedure to upper bound the
processing time of a join, and an algorithm to match this upper
bound. Our procedure improves the runtime exponent $x$ relative to existing work, for a large family of joins.
Moreover, for a class of graphs that we call
$1$-series-parallel graphs,\footnote{A $1$-series-parallel graph
  consists of a source vertex $s$, a target vertex $t$, and a set of
  paths of any length from $s$ to $t$, which do not share any nodes
  other than $s$ and $t$.} we completely resolve the subquadratic
question in the following sense: For each $1$-series-parallel graph,
we can either solve it in subquadratic time, or we show that it cannot
be solved subquadratically unless the $3$-SUM problem~\cite{BDP08} (see Appendix~\ref{sec:3-SUM}) can
be solved in subquadratic time. Note that $1$-series-parallel graphs have fhw equal to $2$. Hence, they can all be solved in
quadratic time using existing algorithms; making our $3$-SUM based lower bound tight. There is a known $3$-SUM based lower bound of $N^{\frac{4}{3}}$ on triangle join processing, which only has a matching upper bound under the assumption that
the matrix multiplication exponent $\omega = 2$. In contrast, our quadratic lower bound can be matched by existing algorithms 
without any assumptions on $\omega$. To our knowledge, this makes it
the only known tight bound on join processing time for small output sizes. 

We also recover our sequential join results within the well-known GHD framework~\cite{gottlob:ght}. We do this using a novel notion of width, which we call $m$-width, that is no larger than fhw, and sometimes smaller than submodular width~\cite{Marx:2013:THP:2555516.2535926} (see Appendix~\ref{sec:widths-comparison}). While we resolve the subquadratic problem on $1$-series-parallel graphs, the general subquadratic problem remains open. We show that known notions of widths, such as submodular width and $m$-width do not fully characterize subquadratically solvable joins (see Appendix~\ref{sec:submodular-width-lower-bound}).

\stitle{Joins on MapReduce:}
Degree information can also be used to improve the efficiency of joins
on MapReduce. Previous work by Beame et al.~\cite{beame:skew} uses
knowledge of heavy hitters (values with high degree) to improve parallel
join processing on skewed data. It allows a limited range of
parallelism (number of processors $p \leq \sqrt{\IN}$), but
subject to that achieves optimal communication for $1$-round MapReduce
algorithms. We use degree information to allow all levels of
parallelism ($p \geq 1$) while processing the join. We also obtain an improved
degree-based upper bound on output size that can be significantly
better than the AGM bound even on simple queries. Our improved
parallel algorithm takes three rounds of MapReduce, matches our
improved bound, and out-performs the optimal $1$-round algorithm
in several cases. As an example, our improved bound lets us
correctly upper bound the output of a sparse triangle join (where each
value has degree $O(1)$) by $\IN$ instead of
$\IN^{\frac{3}{2}}$ as suggested by the AGM bound. Moreover,
we can process the join at maximum levels of parallelism (with each
processor handling only $O(1)$ tuples) at a total communication cost
of $O(\IN)$; in contrast to previous work which requires $\theta(\IN^{\frac{3}{2}})$
communication. Furthermore, previous work~\cite{beame:skew} uses edge
packings to bound the communication cost of processing a join. Edge
packings have the paradoxical property that adding information on the
size of subrelations by adding the subrelations into the join can
make the communication cost larger. As an example suppose a join has a relation $R$, with an attribute $A$ in its schema. Adding $\pi_{A}(R)$ to the 
set of relations to be joined does not change the join output. However, adding a weight term for subrelation $\pi_{A}(R)$ in the edge packing linear program increases its communication cost bound. In contrast, if we add $\pi_{A}(R)$ into the join, our degree based bound does not increase, and will in fact decrease if $|\pi_{A}(R)|$ is small enough.

\stitle{Computing Degree Information:}
In some cases, degree information is not available beforehand or is
out of date. In such a case, we show a simple way to compute the degrees of all
values in time linear in the input size. Moreover, the degree
computation procedure can be fully parallelized in
MapReduce. 
Even after including the
complexity of computing degrees, our algorithms outperform state
of the art join algorithms.

Our paper is structured as follows:
\squishlist
\item In Section~\ref{sec:related-work}, we describe related work. 
\item In Section~\ref{sec:degree-uniformization-counting}, we 
  describe a process called {\em degree-uniformization},
  which mitigates skew. We show the $\MO$ bound on join output size
  that strengthens the exponent in the AGM bound, and describe a method to compute the degrees of all attributes in all relations. 
  
\item In Section~\ref{sec:sequential-join}, we present DARTS, our sequential algorithm that achieves tighter runtime exponents than state-of-the-art. We use DARTs to process several joins in subquadratic time. Then we establish a quadratic runtime lower bound
  for a certain class of queries modulo the $3$-SUM problem. Finally we recover the results of DARTS within the familiar GHD framework, using a novel notion of width ($m$-width) that is tighter than fhw.
\item In Section~\ref{sec:nprr-parallelization-improved}, we present another bound with a tighter exponent than AGM (the DBP bound), and a tunable parallel algorithm whose communication cost at maximum parallelism equals the input size plus the DBP bound. The algorithm's guarantees work on all inputs independent of skew.
\squishend

\section{Related Work}\label{sec:related-work}
We divide related work into four broad categories:

\stitle{New join algorithms and implementation:} The AGM bound~\cite{atserias:agm} is tight on the output size of a multiway join in terms of the query structure and sizes of relations in the query. Several existing join algorithms, such as NPRR~\cite{Ngo:2012:WOJ:2213556.2213565}, LFTJ~\cite{DBLP:journals/corr/abs-1210-0481}, and Generic Join~\cite{ngo:survey}, have worst case runtime equal to this bound. However, there exist instances of relations where the output size is significantly smaller than the worst-case output size (given by the AGM bound), and the above algorithms can have a higher cost than the output size. We demonstrate a bound on output size that has a tighter exponent than the AGM bound by taking into account information on degrees of values, and match it with a parallelizable algorithm.

On $\alpha$-acyclic queries, Yannakakis' algorithm~\cite{Yannakakis81} is instance optimal up to a constant multiplicative factor. That is, its cost is $O(\IN + \OUT)$ where $\IN$ is the input size. For cyclic queries, we can combine Yannakakis' algorithm with the worst-case optimal algorithms like NPRR to get a better performance than that of NPRR alone. This is done using Generalized Hypertree decompositions (GHDS)~\cite{gottlob:ght,chekuri:conjunctive} of the query to answer the query in time $O(\IN^{\textsf{fhw}} + \OUT)$ where fhw is a measure of cyclicity of the query. A query is $\alpha$-acyclic if and only if its fhw is one.
Our work allows us to obtain a tighter runtime exponent than fhw by dealing with values of different degrees separately.

\stitle{Parallel join algorithms:} The Shares~\cite{afrati:mrjoins} algorithm is the optimal one round algorithm for skew free databases, matching the lower bound of Beame et al.~\cite{beame:communication}. But its communication cost can be much worse than optimal when skew is present. Beame's work~\cite{beame:skew} deals with skew and is optimal among $1$-round algorithms when skew is present. The GYM~\cite{DBLP:journals/corr/AfratiJRSU14} algorithm shows that allowing $\log(n)$ rounds of MapReduce instead of just one round can significantly reduce cost. Allowing $n$ rounds can reduce it even further. 
Our work shows that merely going from one to three rounds can by itself significantly improve on existing $1$-round algorithms. Our parallel algorithm can be incorporated into Step $1$ of GYM as well, thereby reducing its communication cost.

\stitle{Using Database Statistics:} The cycle detection algorithm by Alon, Yuster and Zwick~\cite{Alon:1994:FCG:647904.739463} can improve on the fhw bound by using degree information in a sequential setting. Specifically, the fhw of a cycle is two but the AYZ algorithm~\cite{Alon:1994:FCG:647904.739463} can process a cycle join in time $O(\IN^{2-\epsilon} + \OUT)$ where $\epsilon > 0$ is a function of the cycle length. We generalize this, obtaining subquadratic runtime for a larger family of graphs, and develop a general procedure for upper bounding the cost of a join by dealing with different degree values separately.

Beame et al.'s work~\cite{beame:skew} also uses degree information for parallel join processing. Specifically, it assumes that all heavy hitters (values with high degree) and their degrees are known beforehand, and processes them separately to get optimal $1$-round results.
Their work uses edge packings to bound the cost of their algorithm. Edge packings have the counterintuitive property that adding more constraints, or more information on subrelation sizes, can worsen the edge packing cost. This suggests that edge packings alone do not provide the right framework for taking degree information into account. Our work remedies this, and the performance of our algorithm improves when more constraints are added. In addition, Beame et al.~\cite{beame:skew} assume that $M > p^2$ where $M$ is relation size and $p$ is the number of processors. Thus, their algorithm cannot be maximally parallelized. In contrast, our algorithm can work at all levels of parallelism, ranging from one in which each processor gets only $O(1)$ tuples to one in which a single processor does all the processing. 

\stitle{Degree Uniformization:} The partitioning technique of Alon et al.~\cite{Alon:2007:PMS:1238706.1238720} is similar to our {\em degree-uniformization} technique, but has stronger guarantees at a higher cost. It splits a relation into `parts' where the maximum degree of any attribute set $A$ in each part $P$ is within a constant factor of the average degree of $A$ in $P$. In contrast, degree-uniformization lets us upper bound the maximum degree of $A$ in $P$ in absolute terms, but not relative to the average degree of $A$ in $P$. 

Marx's work~\cite{Marx:2013:THP:2555516.2535926} uses a stronger partitioning technique to fully characterize the fixed-parameter tractability of joins in terms of the {\em submodular width} of their hypergraphs. Marx achieves degree-uniformity within all small projections of the output, while we only achieve uniform degrees within relations. Marx's preprocessing is expensive; the technique as written in Section 4 of his paper~\cite{Marx:2013:THP:2555516.2535926} takes time $\Omega(\IN^{2c})$ where $c$ is the submodular width of the join hypergraph. This preprocessing is potentially more expensive than the join processing itself. Our algorithms run in time $O(\IN^{\MW})$ with $\MW < c$ for several joins. Marx did not attempt to minimize this exponent, as his application was concerned with fixed parameter tractability. We were unable to find an easy way to achieve $O(\IN^{c})$ runtime for Marx's technique.

\section{Degree Uniformization}\label{sec:degree-uniformization-counting}
We describe our algorithms for degree-uniformization and counting, as well as
our improved output size bound. Section~\ref{sec:notation} introduces our notation. 
Section~\ref{sec:degree-uniformization}
gives a high-level overview of our join algorithms. Then, we describe
the degree-uniformization which is a key step in our algorithms. In
Section~\ref{sec:mo-bound}, we describe the $\MO$ bound, an
upper bound on join output size that has a tighter exponent than the AGM bound. We
provide realistic examples in which the $\MO$ bound is much
tighter than the AGM bound. Finally, in
Section~\ref{sec:degree-computation} we describe a linear time algorithm for computing degrees.

\subsection{Preliminaries and Notation}\label{sec:notation} 
Throughout the paper we consider a multiway join. Let $\mR$ be the set of relations in the join and $\mA$ be the set of all attributes in those relations' schemas. For any relation $R$, we let $\attr(R)$ denote the set of attributes in the schema of $R$. We wish to process the join $\Join_{R \in \mR} R$, defined as the set of tuples $t$ such that $\forall R \in \mR : \pi_{\attr(R)}(t) \in R$. $|R|$ denotes the number of tuples in relation $R$. For any set of attributes $A \subseteq \mA$, a {\em value} in attribute set $A$ is defined as a tuple from $\bigcup_{R \in \mR : A \subseteq \attr(R)} \pi_{A}(R)$. For any $A \subseteq \attr(R)$, the {\em degree} of a value $v$ in $A$ in relation $R$ is given by the number of times $v$ occurs in $R$ i.e. $\degree(v, R, A) = |\left\lbrace t \in R \mid \pi_{A}(t) = v\right\rbrace|$. For all values $v$ of $A$ in $R$, we must have $\degree(v, R, A) \geq 1$.

In Section~\ref{sec:sequential-join}, we denote a join query with a hypergraph $G$; the vertices in the graph correspond to attributes and the hyperedges to relations. We use $R(X_1, X_2,\ldots, X_k)$ to denote a relation $R$ having schema $(X_1, X_2,\ldots, X_k)$. $\IN$ denotes the input size i.e. sum of sizes of input relations, while $\OUT$ denotes the output size. Our output size bounds, computation costs, and communication costs will be expressed using $O$ notation which hides polylogarithmic factors i.e. $\log^{c}(\IN)$, for some $c$ not dependent on number of tuples $\IN$ (but possibly dependent on the number of relations/attributes). All ensuing logarithms in the paper, unless otherwise specified, will be to the base $\IN$.

\stitle{AGM Bound:}
Consider the following linear program:
\begin{lp} 
$$\textrm{Minimize} \sum_{R \in \mR} w_R\log(|R|) \textrm{ such that } \forall a \in \mA : \sum_{R \in \mR : a \in \attr(R)} w_R \geq 1$$
\end{lp}
A valid assignment of weights $w_R$ to relation $R$ in the linear program is called a {\em fractional cover}. If $\rho*$ is the minimum value of the objective function, then the AGM bound on the join output size is given by $\IN^{\rho*}$. In general, for any set of relations $\mR$, we use $\AGM(\mR)$ to denote the AGM bound on $\Join_{R \in \mR} R$.

\subsection{Degree Uniformization}\label{sec:degree-uniformization}
\begin{algorithm}
\KwIn{Set of relations $\mR$, Bucket range parameter $L$}
\KwOut{$\Join_{R \in \mR} R$}
$1$. Compute $\degree(v, R, A)$ for each $R \in \mR, A \subseteq \attr(R), v \in \pi_{A}(R)$ 

$2$. Compute the set of all $L$-degree configurations $\mC_L$

\ForEach{$c \in \mC_L$}{
$3.1$. Compute partition $R(c)$ of each relation $R$

$3.2$. Compute $\mR(c) = \left\lbrace R(c) \mid R \in \mR \right\rbrace$

$4$. Compute join $J_c = \Join_{R \in \mR(c)} R$
}

$5$. \Return $\bigcup_{c \in \mC_L} J_c$ 
\caption{High level join algorithm} \label{algo:degree-uniformization}
\end{algorithm}

We describe our high level join procedure in Algorithm~\ref{algo:degree-uniformization}. In Step $1$, we compute the degree of each value in each attribute set $A$, in each relation $R$. If the degrees are available beforehand, due to being maintained by the database, then we can skip this step. We further describe this step in Section~\ref{sec:degree-computation}. 

Steps $2, 3$ together constitute {\em degree-uniformization}. In these steps, we partition each relation $R$ by degree. In particular, we assign each value in a relation to a bucket based on its degree: with one bucket for degrees in $[1, L)$, one for degrees in $[L, L^2)$, and so on. Then we process the join using one partition from each relation, for all possible combinations of partitions. Each such combination is referred to as a {\em degree configuration}. We use $c$ to denote any individual degree configuration, $\mC_L$ to denote the set of all degree configurations, $R(c)$ to denote the part of relation $R$ being joined in configuration $c$, and $\mR(c)$ to denote $\left\lbrace R(c) \mid R \in \mR\right\rbrace$. Step $2$ consists of enumerating all degree configurations, and Step $3$ consists of finding the partition of each relation corresponding to each degree configuration.

In Step $4$, we compute $J_c = \Join_{R \in \mR(c)} R$ for each degree configuration $c$. Section~\ref{sec:sequential-join} describes how to perform Step $4$ in a sequential setting, while Section~\ref{sec:nprr-parallelization-improved} describes it for a MapReduce setting. Step $5$ combines the join outputs for each $c$ to get the final output. 

Steps $1$, $2$, $3$ and $5$ can be performed efficiently in MapReduce as well as sequential settings; thus the cost of Algorithm~\ref{algo:degree-uniformization} is determined by Step $4$. Step $4$ is carried out differently in sequential and MapReduce settings. Its cost in the sequential setting is lower than the cost in a MapReduce setting. Steps $1$, $2$, and $3$ have a cost of $O(\IN)$, while Step $5$ has cost $O(\OUT)$. Since reading the input and output always has a cost of $O(\IN + \OUT)$, the only extra costs we incur are in Step $4$ when we actually process the join. Costs for Step $4$ will be described in Sections~\ref{sec:sequential-join} and \ref{sec:nprr-parallelization-improved}.

\stitle{Degree-uniformization:} Now we describe degree-uniformization in detail. We pick a value for a parameter $L$ which we call `bucket range', and define buckets $B_l = [L^l, L^{l+1})$ for all $l \in \mathbb{N}$. Let $\mB = \left\lbrace B_0, B_1,\ldots, \right\rbrace$. For any two buckets $B_i, B_j \in \mB$, we say $B_i \leq B_j$ iff $i \leq j$. A degree configuration specifies a unique bucket for each relation and set of attributes in that relation. Formally:
\begin{definition}
Given a parameter $L$, we define a degree configuration $c$ to be a function that maps each pair $(R,A)$ with $R \in \mR, A \subseteq \attr(R)$ to a unique bucket in $\mB$ denoted $c(R,A)$, such that
$$\forall R, A, A' : A' \subseteq A \subseteq \attr(R) \Rightarrow c(R,A) \leq c(R,A')$$
$$\forall R : c(R,\attr(R)) = B_0 \text{ and } c(R, \emptyset) = B_{\lfloor \log_L(|R|) \rfloor}$$
\end{definition}
\begin{example}
If a join has relations $R_1(X,Y), R_2(Y)$, then a possible configuration is $(R_1, \emptyset) \mapsto B_3$, $(R_1, \left\lbrace X \right\rbrace) \mapsto B_1$, $(R_1, \left\lbrace Y \right\rbrace) \mapsto B_2$, $(R_1, \left\lbrace X, Y \right\rbrace) \mapsto B_0$, $(R_2, \emptyset) \mapsto B_1$, $(R_2, \left\lbrace Y \right\rbrace) \mapsto B_0$.
\end{example}

\begin{definition}\label{def:degree-configuration-partition}
Given a degree configuration $c$ for a given $L$, and a relation $R \in \mR$, we define $R(c)$ to be the set of tuples in $R$ that have degrees consistent with $c$. Specifically:
$$R(c) = \left\lbrace t \in R \mid \forall A \subseteq \attr(R) :  \degree(\pi_{A}(t),R,A) \in c(R,A) \right\rbrace$$
We define $\mC_L$ to be the set of all degree configurations with parameter $L$.
\end{definition}

\begin{example}
For a tuple $(a,b) \in R$, where $L^2 \leq |R| < L^3$, with the degree of $a$ in $B_1$, and that of $b$ in $B_2$, the tuple would be in $R(c)$ if $c(R, \emptyset) = B_2, c(R, \left\lbrace A \right \rbrace) = B_1, c(R, \left\lbrace B \right \rbrace) = B_2, c(R, \left\lbrace A, B \right \rbrace) = B_0$. On the other hand, it would not be in $R(c)$ if
$c(R, \left\lbrace A \right \rbrace) = B_0$, even if we had $c(R,\left\lbrace A, B \right\rbrace) = B_0$, $c(R,\left\lbrace B\right\rbrace) = B_2$.
\end{example}

\noindent
A degree configuration also bounds degrees of values in sub-relations, as stated below:
\begin{lemma}\label{lemma:implicit-degrees}
For all $R \in \mR, A' \subseteq A \subset \attr(R), L > 1, c \in \mC_L, v \in \pi_{A'}(R), j \geq i \geq 0$:
$$c(R, A) = B_i \land c(R, A') = B_j \Rightarrow \degree(v, \pi_{A}(R(c)) ,A') \leq L^{j+1-i}$$
\end{lemma}

\stitle{Choosing $L$: }
The optimal value of parameter $L$ depends on our application. $L$ has three effects : (i) For the DBP/$\MO$ bounds (Sections~\ref{sec:mo-bound},~\ref{sec:nprr-parallelization-improved}) and sequential algorithm (Section~\ref{sec:sequential-join}), the error in output size estimates is exponential in $L$ (with the exponent depending only on the number of attributes) (ii) The load per processor for the parallel algorithm (Section~\ref{sec:nprr-parallelization-improved}) is $O(L)$ (iii) the number of rounds for the parallel algorithm is $\log_L(\IN)$. As a result, we choose a small $L( = 2)$ for the sequential algorithm and DBP/$\MO$ bounds, and a larger $L$ $( = \text{load capacity} = \IN^\gamma$ for some $\gamma  < 1$) for the parallel algorithm.

\subsection{Beyond AGM : The $\MO$ Bound}\label{sec:mo-bound}
We now use degree-uniformization to tighten our upper bound on join output size. 

\begin{definition}
Let $\mR$ be a set of relations, with attributes in $\mA$. For each $R \in \mR, A \subseteq \attr(R)$, let $d_{R,A} = \mathrm{max}_{v \in \pi_{A}(R)} \degree(v,R,A)$. If $A = \emptyset$ then $d_{R,\emptyset} = |R|$. And for any $A \subseteq B \subseteq \attr(R)$, let $d(A,B,R)$ denote $\log(d_{\pi_{B}(R),A})$. Then consider the following linear program for $L$.
\begin{lp}\label{lp:mo-bound}
\begin{align*}
& \textrm{Maximize } s_{\mA} \textrm{ s. t. } \text{ (i) } s_{\emptyset} = 0 \text{ (ii) } \forall A, B \text{ s.t. } A \subseteq B : s_A \leq s_B \\
& \text{ (iii) } \forall A, B, E, R \text{ s.t. } R \in \mR, E \subseteq \mA, A \subseteq B \subseteq \attr(R) : s_{B \cup E} \leq s_{A \cup E} + d(A,B,R)
\end{align*}
\end{lp}
We define $m_{\mA}$ to be the maximum objective value of the above program. 
\end{definition}

\begin{proposition}\label{prop:mo-bound}
The output size $\Join_{R \in \mR} R$ is in $O(\IN^{m_{\mA}})$.
\end{proposition}
This is proved in Appendix~\ref{sec:mo-mw-proofs}. Intuitively, for any $A \subseteq \mA$, $s_A$ stands for possible values of $\log(|\pi_{A}(\Join_{R \in \mR}R)|)$. This explains the first two constraints (projecting onto the empty set gives size $1$, and the projection size over $A$ is monotone in $A$). For the third constraint, we use the fact that each value in $A$ has at most $\IN^{d(A,B,R)}$ values in $B$, thus each tuple in $\pi_{A \cup E}(\Join_{R \in \mR}R)$ can give us at most $\IN^{d(A,B,R)}$ tuples in $\pi_{B \cup E}(\Join_{R \in \mR}R)$. The linear program attempts to maximize the total output size $(\IN^{s_{\mA}})$ while still satisfying the constraints.

We now define the $\MO$ bound.

\begin{definition}\label{def:mo-bound}
Let $\MO(\mR)$ denote the value $m_{\mA}$ for any join query consisting of relations $\mR$. Then the $\MO$ bound is given by $\sum_{c \in \mC_2} \IN^{\MO(\mR(c))}$. 
\end{definition}

\begin{theorem}\label{thm:mo-agm}
The $\MO$ bound is in $O(\textsf{AGM}(\mR))$.  
\end{theorem}
The constant in the $O()$ notation depends on the number of attributes in the query, but not on the number of tuples. This result is proved in two steps. Theorem~\ref{thm:gep-agm} states that the DBP bound (introduced in Section~\ref{sec:nprr-parallelization-improved}) is smaller than the AGM bound, while Theorem~\ref{thm:mo-bound} implies that the $\MO$ bound is smaller than the DBP bound times a constant. 

\begin{example} 
Let $L = 2$ for this example. Consider a triangle join $R(X,Y) \Join S(Y,Z) \Join T(Z,X)$. Let $|R| = |S| = |T| = N$. The AGM bound on this is $N^{3/2}$. Let the degree of each value $x$ in $X$ in both $R$ and $T$ be $h$. For different values of $h$ we will find an upper bound on $m_{\{X,Y,Z\}}$ and hence on the output size. 

{\bf Case 1. $h < \sqrt{N}$:} Then $s_{\{X\}}$ $\leq$ $s_{\emptyset}$ $+$ $d(\emptyset,\{X\},R)$ $=$ $\log(N/h)$. Thus, $s_{\{X,Y\}}$ $\leq$ $s_{\{X\}}$ $+$ $d(\{X\},\{X,Y\},R)$ $\leq$ $\log(N/h)$ $+$ $\log(h)$ $=$ $\log(N)$. Finally, $s_{\{X,Y,Z\}}$ $\leq$ $s_{\{X,Y\}}$ $+$ $d(\{X\}, \{X,Z\}, T)$ $\leq$ $\log(N) + \log(h)$. Thus the $\MO$ bound is $\leq Nh < N^{3/2}$.

{\bf Case 2. $h > \sqrt{N}$:} Since there can be at most $N/h$ distinct $X$ values, we have $d(\{Y\},\{X,Y\},R)$ $\leq$ $\log(N/h))$. More if the degree of $Y$ in $S$ in a degree configuration is $g$, then $s_{\{Y,Z\}}$ $\leq$ $s_{\{Y\}}$ $+$ $d(\{Y\},\{Y,Z\},S)$ $\leq$ $\log(N/g)$ $+$ $\log(g)$ $=$ $\log(N)$. Finally, $s_{\{X,Y,Z\}}$ $\leq s_{\{Y,Z\}}$ $+$ $d(\{Y\},\{X,Y\},R)$ $\leq$ $\log(N)$ $+$ $\log(N/h)$ $=$ $\log(N^2/h)$ $<$ $N^{3/2}$.

The $\MO$ bound has a strictly smaller exponent than AGM unless $h \approx \sqrt{N}$. Computing the AGM bound individually over each degree configuration does not help us do better, as the above example can have all tuples in a single degree configuration.
\end{example}

\begin{example}
Consider a matching database~\cite{beame:communication}, where each attribute has the same domain of size $N$, and each relation is a matching. Thus each value has degree $1$, and $d(A,B,R)$ equals $0$ when $A \neq \emptyset$ and $1$ if $A = \emptyset$. The $\MO$ bound on such a database trivially equals $N$, which can have an unboundedly smaller exponent than the AGM bound.
\end{example}

Appendix~\ref{sec:dbp-agm-examples} similarly compares the DBP and AGM bounds, showing that DBP (and hence $\MO$) has a strictly smaller exponent than AGM for `almost all' degrees.

\subsection{Degree Computation}\label{sec:degree-computation}
If we do not know degrees in advance we can compute them on the fly, as stated below:
\begin{lemma}
Given a relation $R$, $A \subseteq \attr(R)$, and $L > 1$, we can find $\degree(v,R,A)$ for each $v \in \pi_{A}(R)$ in a MapReduce setting, with $O(|R|)$ total communication, in $O(\log_L(|R|))$ MapReduce rounds, and at $O(L)$ load per processor. In a sequential setting, we can compute degrees in time $O(|R|)$.
\end{lemma}
The proof of this lemma is relatively straightforward and can be found in Appendix~\ref{proof:degree-computation}.

To perform degree-uniformization, we compute degrees for all relations $R$, and all $A \subseteq \attr(R)$. The number of such $(R,A)$ pairs is exponential in the number and size of relations, but is still constant with respect to the input size $\IN$. 

\section{Sequential Join Processing}\label{sec:sequential-join}
We present our results on sequential join processing. Section~\ref{sec:sequential-setting} describes our problem setting. In Section~\ref{sec:darts-algorithm} we present our sequential join algorithm, {\em DARTS} (for {\bf D}egree-based {\bf A}ttribute-{\bf R}elation {\bf T}ransform{\bf s}). DARTS handles queries consisting of a join followed by a projection. A join alone is simply a join followed by projection onto all attributes. We pre-process the input by performing degree-uniformization, and then run DARTS on each degree configuration. DARTS works by performing a sequence of {\em transforms} on the join problem; each transform reduces the problem to smaller problems with fewer attributes or relations. We describe each of the transforms in turn. We then show that DARTS can be used to recover (while potentially improving on) known join results such as those of the NPRR algorithm, Yannakakis' algorithm, the fhw algorithm, and the AYZ algorithm. 

In Section~\ref{sec:subquadratic-joins}, we apply DARTS to the subquadratic joins problem; presenting cases in which we can go beyond existing results in terms of the runtime exponent. For a family of joins called $1$-series-parallel graphs, we obtain a full dichotomy for the subquadratic joins problem. That is, for each $1$-series-parallel graph, we can either show that DARTS processes its join in subquadratic time, or that no algorithm can process it in subquadratic time modulo the $3$-SUM problem. Note that $1$-series-parallel graphs have treewidth $2$, making them easily solvable in quadratic time. Thus, our $3$-SUM based quadratic lower bound on some of the graphs is tight making it, to our knowledge, the only tight bound for join processing time with small output sizes. In contrast, there is a $N^{\frac{4}{3}}$ lower bound (using $3$-SUM) for triangle joins, but its matching upper bound depends on the additional assumption that the matrix multiplication exponent equals two. 

In Section~\ref{sec:m-width}, we show that most results of the DARTS algorithms can be recovered using the well known framework of Generalized Hypertree Decompositions (GHDs), along with a novel notion of width we call $m$-width. We show that $m$-width is no larger than fhw, and sometimes smaller than submodular width.

\subsection{Setting}\label{sec:sequential-setting}
In this section, we focus on a sequential join processing setting. We are especially interested in the subquadratic joins problem stated below:
\begin{problem}\label{prob:subquadratic-joins}
For any graph $G$, we let each node in the graph represent an attribute and each edge represent a relation of size $N$. Then we want to know, for what graphs $G$ can we process a join over the relations in {\em subquadratic} time, i.e. $O(N^{2-\epsilon} + \OUT)$ for some $\epsilon > 0$?
\end{problem}
Performing a join in subquadratic time is especially important when we have large datasets being joined, and the output size is significantly smaller than the worst case output size. Note that we define subquadratic to be a $\textrm{poly}(N)$ factor smaller than $N^2$, so for instance a $\frac{N^2}{\log N}$ algorithm is not subquadratic by our definition. 

As an example, if a join query is $\alpha$-acyclic, then Yannakakis' algorithm can answer it in time $O(N + \OUT)$, which is subquadratic. More generally, if the fractional hypertree width (fhw) of a query is $\rho*$, the join can be processed in time $O(N^{\rho*} + \OUT)$ using a combination of the NPRR and Yannakakis' algorithms. The fhw of an $\alpha$-acyclic query is one. For any graph with $\textsf{fhw} < 2$, we can process its join in subquadratic time. The AYZ algorithm (described in Appendix~\ref{sec:AYZ}) allows us to process joins over length $n$ cycles in time $O(N^{2 - \frac{1}{1 + \lceil \frac{n}{2} \rceil}} + \OUT)$, even though cycles of length $\geq 4$ have $\textsf{fhw} = 2$. To the best of our knowledge, this is the only previous result that can process a join with fhw $\geq 2$ in subquadratic time.

The DARTS algorithm is applicable to any join-project problem and not just those with equal relation sizes like in Problem~\ref{prob:subquadratic-joins}. Applying DARTS to Problem~\ref{prob:subquadratic-joins} lets us process several joins in subquadratic time despite having fhw $\geq 2$. Section~\ref{sec:m-width} recovers the subquadratic runtimes of DARTS using GHDs that have $m$-width $< 2$.

\subsection{The DARTS algorithm} \label{sec:darts-algorithm}
We now describe the DARTS algorithm. The problem that DARTS solves is more general than a join. It takes as input a set of relations $\mR$, and a set of attributes $\mO$ (which stands for {\bf O}utput), and computes $\pi_{\mO}\Join_{R \in \mR}R$. When $\mO = \mA$, the problem reduces to just a join. We first pre-process the inputs by performing degree-uniformization. Then each degree configuration is processed separately by DARTS. The $L$ parameter for degree-uniformization is set to be very small ($O(1)$). The total computation time is the sum of the computation times over all degree configurations. Let $G = (c, \mR(c), \mO)$. That is, $G$ specifies the query relations, output attributes, and degrees for each attribute set in each relation according to the degree configuration. We let $c_G, \mR_G, \mO_G$ denote to degree configuration of $G$, the relations in $G$, and the output attributes of $G$. We define two notions of runtime complexity for the join-project problem on $G$:

\begin{definition}
$Q(G)$ is the smallest value such that a join-projection with query structure, degrees, and output attributes given by those in $G$ can be processed in
  time $O(Q(G) + \OUT)$. $P(G)$ is the smallest value such
  that a join-projection with query structure, degrees, and output attributes given
  by those in $G$ can be processed in time $O(P(G))$.
\end{definition}

\begin{example}
As an example of the difference between $P$ and $Q$, consider a chain join $G$ with relations $R_1(X_1,X_2)$, $R_2(X_2,X_3)$, $R_3(X_3,X_4)$, and $\mO = \left\lbrace X_1, X_2, X_3, X_4 \right\rbrace$. All relations have size $N$, and the degree of each attribute in each relation is $\sqrt{N}$. Then $P(G)$ would be $N^2$, the worst case size of the output (where all attributes have $\sqrt{N}$ values and each relation is a full cartesian product). $Q(G)$ on the other hand would be $N$ because the join is $\alpha$-acyclic, and Yannakakis' algorithm lets us process the join in time $O(N + \OUT)$.
\end{example}

\subsubsection{Heavy, Light and Split}
The DARTS algorithm performs a series of {\em transforms} on $G$, each of which reduces it to a smaller problem. In each step, it chooses one of three types of transforms, which we call {\em Heavy}, {\em Light} and {\em Split}. Each transform takes as input $G$ itself and either an attribute or a set of attributes in the relations of $G$. Then it reduces the join-project problem on $G$ to a simpler problem via a {\em procedure}. This reduction gives us a {\em bound} on $P(G)$ and/or $Q(G)$ in terms of the $P$ and $Q$ values of simpler problems. We describe each of these transforms in turn, along with their input, procedure, and bound.

\paragraph*{Heavy:} 
{\bf Input:} $G$, An attribute $X$ 

\noindent
{\bf Procedure:} Let $\mR_X = \left\lbrace R \in \mR(c) \mid X \in \attr(R) \right\rbrace$. Then we compute the values of $x \in X$ that lie in all relations in $\mR_X$ i.e. $\textsf{vals}(X) = \bigcap_{R \in \mR_X} \pi_{X}R$. Then for each $x \in \textsf{vals}(X)$, we marginalize on $x$. That is, we solve the {\em reduced problem}:
$$J_x = \pi_{\mO \setminus \left\lbrace X \right\rbrace} \left(\Join_{R \in (\mR(c) \setminus \mR_X)} R \Join_{R \in \mR_X} (\pi_{\mA \setminus \left\lbrace X \right\rbrace }\sigma_{X=x}R)\right)$$ 
Our final output is $\bigcup_{x \in \textsf{vals}(X)} (\pi_{\mO} x) \times J_x$. For each relation $R \in \mR_X$, let $d_R$ be the maximum value in bucket $c(R,\left\lbrace X \right\rbrace)$. So $|\textsf{vals}(X)| \leq \min_{R \in \mR_X} \frac{|R|}{d_R}$. Secondly, in each reduced problem $J_x$, the size of each reduced relation $\pi_{\mA \setminus \left\lbrace X \right\rbrace }\sigma_{X=x}R$ for $R \in \mR_X$ reduces to at most $d_R$. Let $G'$ denote the reduced relations, degrees, and output attributes for $J_x$. This gives us:

\noindent
{\bf Bound:}
$Q(G) \leq \left(\mathrm{min}_{R \in \mR_X} \frac{|R|}{d_R}\right) Q(G')\text{ , }P(G) \leq \left(\mathrm{min}_{R \in \mR_X} \frac{|R|}{d_R}\right) P(G')$

\paragraph*{Light:} 
{\bf Input:} $G$, An attribute set $X$

\noindent
{\bf Procedure:} The light transform reduces the number of relations in $G$. Define $\mR_X = \left\lbrace R \in \mR(c) \mid \attr(R) \subseteq X \right\rbrace$. We compute $R_X = \Join_{R \in \mR(c)} \pi_{X}R$. This subjoin is computed using a sequential version of the parallel technique in Section~\ref{sec:nprr-parallelization-improved}. Hence it takes time equal to the DBP bound on that join. Then we delete relations in $\mR_X$ from $G$, and add $R_X$ into $\mR_G$. The degrees for attributes in $R_X$ can be computed in terms of degrees in the relations from $\mR_X$. As long as $|\mR_X| > 1$, this gives us a reduced problem $G'$. $\mO$ stays unchanged for the reduced problem. The size of relation $R_X$ can be upper bounded using the DBP bound as well. Let $\DBP(G,X)$ denote this bound.

\noindent
{\bf Bound:} $Q(G) \leq \DBP(G,X) + Q(G')\text{ , }P(G) \leq \DBP(G,X) + P(G')$


\paragraph*{Split:} 
{\bf Input:} $G$, An {\em articulation set} $S$ of attributes~\cite{Gross:2013:HGT:2613412} such that there are joins $G_1, G_2$ whose attribute sets have no attribute outside $S$ in common, and $\mR_G \subseteq \mR_{G_1} \cup \mR_{G_2}$. Also, $S$ satisfies either (i) $S \subseteq \mO$, or (ii) $\mO \subseteq \bigcup_{R \in \mR_{G_2}} \attr(R)$.

\noindent
{\bf Procedure:} We compute $R_S = \pi_{S} \left(\Join_{R \in \mR_{G_1}} R \right)$. This takes time $P(G'_1)$, where $G'_1$ is like $G_1$ but with $\mO_{G'_1} = S$. Let $J_2 = \left(\Join_{R \in \mR_{G_2}} R\right) \Join R_S$.
If $\mO \subseteq \bigcup_{R \in \mR_{G_2}} \attr(R)$, then we compute and output $\pi_{\mO} J_2$, and we are done. This step costs $P(G_2)$.
Otherwise, $S \subseteq \mO$. We compute $O_2 = \pi_{\mO} J_2$. Each tuple in $O_2$ has a matching output tuple for $G$. Then we set $R_S = R_S \cap \pi_{S} O_2$ and compute $O_1 = \pi_{\mO}(\Join_{R \in \mR_{G_1}} R \Join R_S)$. Then for each tuple $t \in R_S$, we take each pair of matching tuples $t_1 \in O_1$, $t_2 \in O_2$ and output $t_1 \Join t_2$. Let $G''_1$ be like $G_1$, but with $\mO_{G''_1} = \mO \cap \left(\bigcup_{R \in \mR_{G_1}} \attr(R)\right)$, and $G''_2$ be defined similarly. This gives us:

\noindent 
{\bf Bound:} If $S \subseteq \mO$, then $Q(G) \leq P(G'_1) + Q(G''_1) + Q(G''_2)$ \\
If $\mO \subseteq \bigcup_{R \in \mR_{G_2}} \attr(R)$, then $P(G) \leq P(G'_1) + P(G_2)$.

\subsubsection{Combining the Transforms}
Once we know the transforms, the DARTS algorithm is quite straightforward. It considers all possible sequences of transforms that can be used to solve the problem, and picks the one that gives the smallest upper bound on $Q(G)$. The number of such transform sequences is exponential in the number of attributes and relations, but constant with respect to data size. The $P$ and $Q$ values of various $G$s can be computed recursively given a degree configuration. The $G'$ obtained in each recursive step itself specifies a degree configuration, over a smaller problem. The degrees in $G'$ can be computed in terms of degrees in $G$. Note that in some cases, we do not have cost bounds available e.g. we do not have a $P$ bound for the Split transform when $S \subseteq \mO$. This is a part of the DARTS algorithm. DARTS only considers performing a transform when it can upper bound the resulting cost.

We show that DARTS can be used to recover existing results on sequential joins. 

\begin{proposition}\label{prop:nprr}
If we compute the join using a single Light transform, our total cost is $\leq$ the AGM bound, thus recovering the result of the NPRR algorithm~\cite{Ngo:2012:WOJ:2213556.2213565}.
\end{proposition}

\begin{proposition}\label{prop:yannakakis}
If we successively apply the Split transform on an $\alpha$-acyclic join, with $G_1$ being an ear of the join in each step, then the total cost of our algorithm becomes $O(\IN + \OUT)$, recovering the result of Yannakakis' algorithm~\cite{Yannakakis81}.
\end{proposition}

\begin{proposition}\label{prop:fhw}
If a query has fractional hypertree width equal to fhw, then using a combination of Split and Light transforms, we can bound the cost of running DARTS by $O(\IN^{fhw} + \OUT)$, recovering the fractional hypertree width result.
\end{proposition}

\begin{proposition}\label{prop:ayz}
A cycle join of length $n$ with all relations having size $N$, can be processed by DARTS in time $O(N^{2 - \frac{1}{1 + \lceil \frac{n}{2}\rceil}} + \OUT)$, recovering the result of the AYZ algorithm~\cite{Alon:1994:FCG:647904.739463}. 
\end{proposition}

These propositions are proved in Appendix~\ref{sec:previous-recovery}. In the next subsection, we present a few of the cases in which we can go {\em beyond} existing results. Since we are primarily interested in joins, the output attribute set $\mO$ below is always assumed to be $\mA$. 

\subsection{Subquadratic Joins}\label{sec:subquadratic-joins}
Now we consider applications of DARTS to the subquadratic joins problem. Analyzing a run of DARTS on a join graph allows us to obtain a subquadratic runtime upper bound in several cases. Appendix~\ref{sec:TCS} mentions a simple extension of the AYZ result to graphs that are trees with cycles embedded in them. We now define a set of graphs for which we have a complete decision procedure to determine if they can be solved in subquadratic time modulo the $3$-SUM problem.
 
\stitle{$1$-series-parallel graphs}
\begin{definition}\label{def:1-series-parallel}
A $1$-series-parallel graph is one that consists of :
\squishlist
\item A source node $X_S$
\item A sink node $X_T$
\item Any number of paths, of arbitrary length, from $X_S$ to $X_T$, having no other nodes in common with each other
\squishend
\end{definition}
Equivalently, a $1$-series-parallel graph is a series parallel graph that can be obtained using any number of series transforms (which creates paths) followed by exactly one parallel transform, which joins the paths at the endpoints. A cycle is a special case of a $1$-series-parallel graph.

\begin{theorem}\label{thm:1-series-parallel-graphs}
For $1$-series-parallel graphs, the following decision procedure determines whether or not the join over that graph can be processed in sub-quadratic time:
\begin{enumerate}
\item If there is a direct edge (path of length one) between $X_S$ and $X_T$, then the join can be processed in sub-quadratic time. Else:
\item Remove all paths of length two between $X_S$ and $X_T$, as they do not affect the sub-quadratic solvability of the join problem. Then
\item If the remaining number of paths (obviously all having length $\geq 3$) is $\geq 3$, then the join cannot be processed in subquadratic time (modulo $3$-SUM). If the number of remaining paths is $< 3$, then the graph can be solved in sub-quadratic time.
\end{enumerate}
\end{theorem}

Theorem~\ref{thm:1-series-parallel-graphs} establishes the decision procedure for subquadratic solvability of $1$-series-parallel graphs. Appendix~\ref{sec:darts-examples} gives an example of a subquadratic solution for a specific $1$-series-parallel graph, namely $K_{2,n}$, followed by an example on the general bipartite graph $K_{m,n}$. In both these examples, DARTS achieves a better runtime exponent than previously known algorithms. We now make three statements that together imply Theorem~\ref{thm:1-series-parallel-graphs}. They are formally stated and proved in Appendix~\ref{sec:1-series-parallel-proofs} (Lemmas~\ref{lemma:s-t-edge},~\ref{lemma:s-t-2-path},~\ref{lemma:s-t-3-paths}).

\begin{itemize}
\item If we have a $1$-series-parallel graph, which has a direct edge from $X_S$ to $X_T$ (i.e. a path of length $1$), then a join on that graph can be processed in subquadratic time.
\item Suppose we have a $1$-series-parallel graph $G$, which does not have a direct edge from $X_S$ to $X_T$, but has a vertex $X_U$ such that there is an edge from $X_S$ to $X_U$ and from $X_U$ to $X_T$ (i.e. a path of length $2$ from $X_S$ to $X_T$). Let $G'$ be the graph obtained by deleting the vertex $X_U$ and edges $X_SX_U$ and $X_UX_T$. Then the join on $G$ can be processed in subquadratic time if and only if that on $G'$ can be processed in subquadratic time.
\item Let $G$ be any $1$-series-parallel graph which does not have an edge from $X_S$ to $X_T$, but has $\geq 3$ paths of length at $\geq 3$ each, from $X_S$ to $X_T$. Then a join over $G$ can be processed in subquadratic time only if the $3$-SUM problem can be solved in subquadratic time.
\end{itemize}

\subsection{A new notion of width ($m$-width)}\label{sec:m-width}
We demonstrate a way to formulate the DARTS algorithm for joins (without projection) in terms of GHDs. 

For each $A \in \mA$, we define $m_A$ similarly to how we defined $m_\mA$ in Section~\ref{sec:mo-bound}. Specifically, for each $A$, we use the same constraints as in linear program~\ref{lp:mo-bound}, but the objective is set to $\textrm{Maximize } s_A$ instead of $\textrm{Maximize } s_{\mA}$. $m_A$ is then defined as the value of this objective function. We let $\textsf{Prog}(A)$ denote the above linear program for finding $m_A$. Then the size $|\pi_{A}(\Join_{R \in \mR} R)|$ must be bounded by $\IN^{m_A}$ for all $A \subseteq \mA$ (see Appendix Proposition~\ref{prop:mo-bound-general}). Moreover, for any GHD $D = (\mT, \chi)$ of query $\mR$, we can define $\MW(D, \mR)$ to be $\mathrm{max}_{t \in \mT}(m_{\chi(t)})$. And $\MW(\mR)$ is simply the minimum value of $\MW(D,\mR)$ over all GHDs $D$. Thus we have:

\begin{definition}\label{m-width}
The $m$-width of a join query $\Join_{R \in \mR} \mR$ (possibly with non-uniform degrees), is given by $\max_{c \in \mC_2} \MW(\mR(c))$.
\end{definition}

\begin{theorem}\label{thm:mw-join}
A query with $m$-width $\MW$ can be answered in time $O(\IN^{\MW} + \OUT)$.
\end{theorem}

This theorem lets us recover all our subquadratic joins results as well. That is, for the $1$-series-parallel graphs that have a subquadratic join algorithm (as per Theorem~\ref{thm:1-series-parallel-graphs}), we can construct a GHD that has $m$-width less than $2$ (see Appendix~\ref{sec:ghd-darts-recovery}). 

We can show the $\MO$ bound to be better than the DBP bound (and consequently, the AGM bound, as stated in Theorem~\ref{thm:mo-agm} earlier).
\begin{theorem}\label{thm:mo-bound}
For any join query $\mR$, and any degree configuration $c \in \mC_2$, $\MO(\mR(c)) \leq \DBP(\mR(c),2) + |C|\log(2)$, where $C$ is the cover used in the DBP bound.
\end{theorem}
Note that since logarithms are to the base $\IN$, the $|C|\log(2)$ term is negligible even though it goes in the exponent of the bound i.e. its exponent is a constant. Theorems~\ref{thm:mw-join} and~\ref{thm:mo-bound} let us recover all the results of the DARTS algorithm (see Appendix~\ref{sec:ghd-darts-recovery}). 

The theorems also imply that our new notion of width ($m$-width) is tighter than fhw. Appendix~\ref{sec:widths-comparison} compares $m$-width to submodular width (which, barring $m$-width, is the tightest known notion of width applicable to general joins). Appendix~\ref{sec:widths-comparison} shows examples where $m$-width is tighter than submodular width, but we do not know in general if $m$-width is tighter than submodular width. 

Appendix~\ref{sec:submodular-width-lower-bound} shows that while $m$-width $<2$ implies subquadratic solvability, the converse is not true; we show an example join which has $m$-width and submodular width $= 2$ but can be solved in subquadratic time. Thus known notions of width do not fully characterize subquadratically solvable graphs.

\section{Parallel Join Processing}\label{sec:nprr-parallelization-improved}
Like in sequential settings, degree-uniformization can be applied in a MapReduce setting. We first present the DBP bound, which is a bound on output size that is tighter than AGM bound (but not tighter than $\MO$), and characterizes the complexity of our parallel algorithm. Then we present a $3$-round MapReduce algorithm whose cost equals the DBP bound at the highest level of parallelism.

\paragraph*{The DBP Bound}
We start by defining a quantity called the Degree-based packing (DBP). 
\begin{definition}\label{def:generalized-packing}
Let $\mR$ be a set of relations, with attributes in $\mA$. Let $C$ denote a cover i.e. a set of pairs $(R, A)$ such that $R \in \mR$, $A \subseteq \attr(R)$, and $\bigcup_{(R,A) \in C} A = \mA$. Let $L > 1$. Then, consider the following linear program for $C, L$.
\begin{lp}\label{lp:degree-packing-primal}
$$\textrm{Minimize } \sum_{a \in \mA} v_a \textrm{ such that } \forall (R,A) \in C , \forall A' \subseteq A : \sum_{a \in A'} v_a \geq \log\left(\frac{d_{\pi_{A}(R), A \setminus A'}}{L}\right)$$
\end{lp}
If $O_{C,L}$ is the maximum objective value of the above program, then we define $\DBP(\mR, L)$ to be $\min_{C} O_{C,L}$ where the minimum is taken over all covers $C$. 
\end{definition}

\begin{proposition}\label{prop:dbp-bound}
Let $L>1$ be a constant. Then the output size of $\Join_{R\in \mR} R$ is in $O(\IN^{\DBP(\mR, L)})$.
\end{proposition}
We implicitly prove this result by providing a parallel algorithm whose complexity equals the output size bound at the maximum parallelism level. We can now define the DBP bound. We arbitrarily set $L=2$ for this definition (choosing another constant value only changes the bound by a constant factor). Thus, we define the {\em DBP bound} to be $\sum_{c \in \mC_2} \IN^{\DBP(\mR(c), 2)}$. As a simple corollary, the output size of the join is $\leq$ the DBP bound.

\begin{theorem}\label{thm:gep-agm}
For each degree configuration $c \in \mC_L$, $\IN^{\DBP(\mR(c), L)} \leq \AGM(\mR(c))$. 
\end{theorem}
We prove this theorem using a sequence of linear program transformations, starting with the AGM bound, and ending with the DBP bound, which each transformation decreasing the objective function value. The key transform is the fifth one, where we switch from a cover-based program to a packing-based program. The proof itself is long and is deferred to Appendix~\ref{proof:gep-agm}. Appendix~\ref{sec:dbp-agm-examples} contains a simple triangle-join example where the DBP bound has a tighter exponent than the AGM bound, and another more general example showing that the DBP bound has a strictly better exponent than AGM for `almost all' degrees. 

\paragraph*{Parallel Join Algorithm} We present our parallel $3$-round join algorithm. The algorithm works at all levels of parallelism specified by load level $L$. Its communication cost matches the DBP bound when $L = O(1)$. We formally state the result, and then provide an example of its performance (with additional examples provided in Appendix~\ref{sec:parallel-algo-examples}).

\begin{theorem}
For any value of $L$, we can process a join in $O(\log_{L}(\IN))$ rounds (three rounds if degrees are already known) with load $O(L)$ per processor and a communication cost of $O(\IN + \OUT + \max_{c \in \mC_L}L \cdot \IN^{\DBP(\mR(c), L)})$.
\end{theorem}
\begin{proof} (Sketch)

The join consists of the following steps:
\begin{enumerate}
\item Perform degree finding and uniformization using bucket range $L$, as shown in Section~\ref{sec:degree-uniformization}.
\item For each degree configuration, re-compute the degrees, and use them to solve Linear Program~\ref{lp:degree-packing-primal} for each cover. Let $C$ be the cover that gives the smallest objective value. This smallest value will equal $\DBP(\mR(c), L)$. 
\item {\bf MapReduce round $1$}: Join all the $\pi_{A}(R) : (R,A) \in C$ in a single step using the shares algorithm. Each attribute $a$ is assigned share $\IN^{v_a}$, where $v_a$ is from our solution to Linear Program~\ref{lp:degree-packing-primal}. This ensures a load of $O(L)$ per processor, and communication cost of $O(\max_{c \in \mC_L}L \cdot \IN^{\DBP(\mR(c), L)})$ (Lemma~\ref{lemma:shares-communication-load}, proved in Appendix).
\item {\bf MapReduce rounds $2-3$}: For each $R$ such that $(R, \attr(R)) \notin C$, semijoin it with the output of the previous join. The semijoins for all such $R$s can be done in parallel in one round, followed by intersection of the semijoin results in the next round. This can be done with $O(1)$ load and communication cost of $O(\IN + \OUT)$.
\end{enumerate}

\begin{lemma}\label{lemma:shares-communication-load}
The shares algorithm, where each attribute $a$ has share $\IN^{v_a}$, where $v_a$ is from the solution to Linear Program~\ref{lp:degree-packing-primal}, has a load of $O(L)$ per processor with high probability, and a communication cost of $O(\max_{c \in \mC_L}L  \cdot \IN^{\DBP(\mR(c), L)})$.
\end{lemma}
\end{proof}

\begin{example}
Consider the sparse triangle join, with $\mR = \left\lbrace R_1(X,Y), R_2(Y,Z), R_3(Z,X) \right\rbrace$. Each relation has size $N$, and each value has degree $O(1)$. When the load level is $L < N$, the join requires $\DBP(\mR, L) = \frac{N}{L}$ processors. Equivalently, when we have $p$ processors, the load per processor is $\frac{N}{p}$, which means it decreases as fast as possible as a function of $p$.

In contrast the vanilla shares algorithm allocates a share of $p^{\frac{1}{3}}$ to each attribute, and the load per processor is $Np^{-\frac{2}{3}}$. Current state of the art work~\cite{beame:skew} has a load of $Np^{-\frac{2}{3}}$ as well.
\end{example}

We further explore and generalize this example in Appendix~\ref{sec:parallel-algo-examples}. We also show an example where our parallel algorithm operating at maximum parallelism still has lower total cost than existing state-of-the-art sequential algorithms.

\section{Conclusion and Future Work}
We demonstrated that using degree information for a join can let us tighten the exponent of our output size bound. We presented a parallel algorithm that works at all levels of parallelism, and whose communication cost matches a tightened bound at the maximum parallelism level. We proposed the question of deciding which joins can be processed in subquadratic time, and made some progress towards answering it. We showed a tight quadratic lower bound for a family of joins, making it the only known tight bound that makes no assumptions about the matrix multiplication exponent. We presented an improved sequential algorithm, namely DARTS, that generalizes several known join algorithms, while outperforming them in several cases. We recovered the results of DARTS in the GHD framework, using a novel notion of width that is tighter than fhw and sometimes tighter than submodular width as well.

We presented several cases in which DARTS outperforms existing algorithms, in the context of subquadratic joins. However, it is likely that DARTS outperforms existing algorithms on joins having higher treewidths as well. A fuller exploration of the improved upper bounds achieved by DARTS is left to future work. Appendix~\ref{sec:submodular-width-lower-bound} shows an example where a join can be performed in subquadratic time despite its $m$-width/submodular width being $= 2$. Thus the problem of precisely characterizing which joins can be performed in subquadratic time remains open. Moreover, we focused entirely on using degree information for join processing; using other kinds of information stored by databases to improve join processing is a promising direction for future work. 

\subparagraph*{Acknowledgements}
\scriptsize
The authors would like to thank Atri Rudra for pointing out the connection to submodular width.
CR gratefully acknowledges the support of the Defense Advanced Research Projects
Agency (DARPA) XDATA Program under No. FA8750-12-2-0335 and DEFT Program under No. FA8750-13-2-0039,
DARPAs MEMEX program under No. FA8750-14-2-0240, the National Science Foundation (NSF) under CAREER
Award No. IIS-1353606, Award No. No. CCF-1356918 and EarthCube Award under No. ACI-1343760, the Office of
Naval Research (ONR) under awards No. N000141210041 and No. N000141310129, the Sloan Research Fellowship,
the Moore Foundation Data Driven Investigator award, and gifts from American Family Insurance, Google, Lightspeed
Ventures, and Toshiba.

\bibliography{Degree}

\appendix
\normalsize

\section{Background}\label{sec:background}
\subsection{Generalized Hypertree Decompositions (GHDs)}
\begin{definition}
Given a set of relations $\mR$ over attributes $\mA$, a \emph{generalized hypertree decomposition} is a pair $(\mT, \chi)$ where $\mT$ is a tree and $\chi$ is a function from nodes of $\mT$ to $2^{\mA}$ such that
\begin{itemize}
\item For each relation $R \in \mR$, there exists a tree node in $\mT$ that covers the relation, i.e. $\attr(R) \subseteq \chi(t)$.
\item For each attribute $A \in \mA$, the set of tree nodes containing $A$ i.e. $\{t \mid A \in \chi(t)\}$ forms a connected subtree.
\end{itemize}
\end{definition}

The latter condition is called the ``running intersection property''. The $\chi(t)$ sets are referred to as `bags' of the GHD. Using GHDs, we can define several notions of `width', which capture the cyclicity of a query. For example, the treewidth of a GHD is the maximum value of $|\chi(t)|-1$ over nodes $t$ in $\mT$, and treewidth of a query is the treewidth of its minimum-treewidth GHD. Similarly, fractional hypertreewidth ($\mathsf{fhw}$) is the maximum value of $\log_{IN}(\AGM(\chi(t)))$ over $t \in \mT$ where $\AGM(\chi(t))$ is the AGM bound over the set of attributes in $\chi(t)$ for the given relations $\mR$. Again the $\mathsf{fhw}$ of a query is the minimum $\mathsf{fhw}$ over its GHDs. 

If the width of a GHD is $w$ (for any of the known notions of width), then the size of the join $\Join_{R \in \mR} \pi_{\chi(t)}(R)$ is $\leq \IN^w$ for all $t \in \mT$. Thus
the join can be computed by first computing the join within the bag as above, and then running Yannakakis' algorithm~\cite{Yannakakis81} on the resulting relations with a runtime of $\IN^{w} + \OUT$. 

\subsection{MapReduce}\label{sec:mapreduce}
In the MapReduce (MR) model, there are unboundedly many processors on a networked file system. Each processor has unbounded hard disk space and load capacity $L$ (explained later). The computation proceeds in two phases.
 
{\bf Step 1}: Each processor (referred to as a mapper), reads its tuples from its hard disk and sends each tuple to one or more processors (called reducers). 
The total number of tuples received by each reducer from all mappers should not exceed load capacity $L$.

{\bf Step 2:} Each reducer locally processes the $\leq L$ tuples it receives, and streams its output to the network file system. The output size at a reducer can exceed load capacity $L$ as it is streamed to the network file system.

The {\em communication cost} of each round is defined as the total number of tuples sent from all mappers to reducers. We measure the complexity of our algorithms in terms of communication cost and number of rounds. 

\subsection{The Shares Algorithm}\label{sec:shares}
Shares is a one-round MapReduce algorithm. Shares is parameterized algorithm, whose communication cost is different for different queries and machine sizes. Suppose we have a join $\Join_{R \in \mR} R$ with attribute set $\mA$. Shares assigns a parameter  $S_A$, called a `share' to each attribute $A \in \mA$. It hashes each attribute $A$ into $S_A$ buckets using a hash function $h_A$. It uses $\Pi_{A \in \mA} S_A$ processors, with one processor corresponding to each tuple of hash values. For any processor $P$ and attribute $A$, we use $P(A)$ to denote the hash value of $A$ corresponding to processor $P$. 

Shares uses a single round of MapReduce. In that round, each tuple $t \in R, R \in \mR$, is sent to every processor $P$ such that $P(A) = h_A(t(A)) \text{ } \forall \text{ } A \in \attr(R)$. Then, each processor joins all the tuples it receives, and the final output of the join equals the union of the outputs produced by all processors.

Each tuple in relation $R$ gets sent to $\Pi_{A \notin \attr(R)} S_A$ processors. The communication cost of this algorithm is thus $\sum_{R \in \mR} |R| \Pi_{A \notin \attr(R)} S_A$. The expected `load' on each processor (number of input tuples it receives) is $\sum_{R \in \mR} |R| \Pi_{A \in \attr(R)} (S_A)^{-1}$, which is simply the total communication divided by the total number of processors. On the other hand, the variance in load can be high, leading to some processors receiving a very high number of input tuples. In general, the shares $S_A$ are chosen so as to minimize the total communication cost, given the number of processors. 

\subsection{Articulation set}\label{sec:articulation-set}
Suppose we have a hypergraph $\mH = (\mV, \mE)$ with $\mE \subseteq 2^{\mV}$. The hypergraph is {\em connected} if for each pair $v_1, v_2 \in \mV$, there exists a sequence $u_0, u_1,\ldots, u_k$ such that $u_i \in \mV \text{ } \forall \text{ } 0 \leq i \leq k$, $u_0 = v_1, u_k = v_2$ and $\forall \text{ } i < k \text{ } \exists \text{ } e \in \mE : u_i \in e \land u_{i+1} \in e$.

If $\mH$ is connected, then an articulation set $S$ is a set $S \subsetneq \mV$ such that the hypergraph $\mH_{-S} = (\mV \setminus S, \{e \setminus S \mid e \in \mE\})$ is {\em not connected}. Equivalently, $S$ is an articulation set if $\exists \text{ } \mV' \subsetneq \mV \setminus S$ such that $\forall \text{ } e \in \mE$, either $e \subseteq \mV' \cup S$ or $e \subseteq \mV \setminus \mV'$.

\subsection{The AYZ algorithm}\label{sec:AYZ}
Consider a join given by $R_1(X_1,X_2)$, $R_2(X_2,X_3)$,$\ldots$ $R_n(X_n,X_1)$, for $n \geq 4$. This is the cycle join of length $n$. The cycle has fhw equal to 2, so the join can be processed in time $O(N^2 + \OUT)$. However, we can even process the join in subquadratic time as follows:
For each attribute $X_j$, we compute the degree of each of its values. We choose a threshold $\Delta$. We call any value with degree less than $\Delta$ {\em light}, and other 
values {\em heavy}. We process heavy and light values separately. The number of heavy values in an attribute can be at most $\frac{N}{\Delta}$. For each heavy value $h$ in each attribute $X_j$, we `marginalize' over the value i.e. restrict $X_j$ to $h$. So effectively we compute the join with all values in $X_j$ other than $h$ removed. This effectively turns the join into a chain join 
$$R_{j+1}R_{j+2}\ldots, R_nR_1R_2\ldots, (\pi_{X_{j-1}X_{j+1}}\sigma_{X_j=h}R_{j-1}R_j)$$
Adding column $X_j = h$ to the output of the chain above gives us the output for $h$. Let us call this output $\OUT_h$. Using Yannakakis' algorithm on the chain lets us solve it in time $O(N + \OUT_h)$. Thus, the total time for processing all heavy values in all attributes is 
\begin{align*}
\sum_h O(N + \OUT_h) &= \sum_h O(N) + \sum_h O(\OUT_h)\\
&= O(\frac{nN}{\Delta}N) + O(\OUT) \\
&= O(\frac{N^2}{\Delta} + \OUT)
\end{align*}
This way, we can find all outputs containing at least one heavy value. After this is done, we can delete all the heavy values, and process only light values. This is done by a simple brute force search. We start with each value in $X_1$, which has at most $\Delta$ neighbors in $X_2, X_n$, which together have at most $\Delta^2$ neighbors in $X_3, X_{n-1}$ and so on. At $X_{\frac{n}{2}}$ we take intersection of neighbors from both directions. The total running time for this procedure is the number of values in $X_1$ i.e. $N$, times the total number of neighbors explored per $X_1$ value, which is $\Delta^{\lceil \frac{n}{2} \rceil}$. Thus, the total processing time of the join is
$$O(\frac{N^2}{\Delta} + N\Delta^{\lceil \frac{n}{2}\rceil} + \OUT)$$
Setting $\Delta = N^{\frac{1}{1 + \lceil \frac{n}{2}\rceil}}$ gives us the minimum value of the running time, which is also subquadratic.

\subsection{$3$-SUM}\label{sec:3-SUM}
We first define the $3$-SUM problem below.
\begin{problem}
\label{prob:3sum}
The \textbf{3SUM} problem : Given $n$ integers $x_1$, $x_2$, ... $x_n$ all polynomial sized in $n$, do there exist three of those numbers, $x_i$, $x_j$, $x_k$ such that $x_i + x_j + x_k = 0$?
\end{problem}
There is no known algorithm for solving this problem in time $O(n^{2-\epsilon})$ for any $\epsilon > 0$, and it is believed that such an algorithm does not exist. On the other hand, there is a known algorithm for solving the problem in time that is smaller than $n^2$ by a subpolynomial (log) factor. We next state the $3$-XOR problem, which is subquadratically reducible from the $3$-SUM problem.
\begin{problem}
\label{prob:3xor}
The \textbf{3XOR} problem : Given $n$ integers $x_1$, $x_2$, ... $x_n$ all polynomial sized in $n$, do there exist three of those numbers, $x_i$, $x_j$, $x_k$ such that $x_i \oplus x_j \oplus x_k = 0$ where $\oplus$ refers to bitwise xor?
\end{problem}

\section{Degree Computation}\label{proof:degree-computation}
\begin{lemma*}
Given a relation $R$ and a $A \subseteq \attr(R)$, and a $L > 1$, we can find $\degree(v,R,A)$ for each $v \in \pi_{A}(R)$ in a MapReduce setting, using $O(|R|)$ total communication, in $O(\log_L(|R|))$ MapReduce rounds, and with $O(L)$ load per processor. In a sequential setting, we can compute degrees in time $O(|R|)$.
\end{lemma*}
\begin{proof}
Suppose the schema of $R$ has $K$ attributes $X_1, X_2,\ldots, X_K$. Let $|A| = K' \leq K$. Without loss of generality, we can assume that $A = \left\lbrace X_1, X_2,\ldots, X_{K'}\right\rbrace$. We want to find the degree of each value in $\pi_{X_1, X_2,\ldots, X_{K'}}(R)$. We make no assumption about the starting location of different tuples of $R$, each tuple of $R$ could be in a different processor. 

We have $|R| \times \frac{|R|}{L}$ processors, indexed by $(k_1, k_2)$, with $1 \leq k_1, k_2 \leq |R|$. For each tuple $(x_1, x_2,\ldots, x_K) \in R$, its processor finds a hash $k_1 \in \left\lbrace 1, 2,\ldots, |R| \right\rbrace$ of $(x_1, x_2,\ldots, x_{K'})$. In addition, the processor generates a random number $k_2 \in \left\lbrace 1, 2,\ldots, \frac{|R|}{L} \right\rbrace$, and sends $(x_1, x_2,\ldots, x_{K'})$ to processor $(k_1, k_2)$. Each processor receives at most $O(L)$ tuples in expectation, because of the second random hash. The first index of the processors ($k_1$) corresponds to the tuple value. Because we have $|R|$ buckets for the first index, each hash value $k_1$ should correspond to $O(1)$ distinct values of $(x_1, x_2,\ldots, x_{K'})$. Each tuple is associated with a `count' field. The initial value of the count field, when the tuple is sent to any processor in the starting step, is $1$. 

The next $\log_L(|R|)$ steps are as follows: In each step, each processor $(k_1, k_2)$ locally aggregates the count of each of its tuples (since a processor may have recieved multiple copies of the same $(x_1, x_2,\ldots, x_{K'})$ value from different processors), and sends each aggregated tuple-count pair to processor $(k_1, \lceil \frac{k_2}{L}\rceil)$. Thus, in $\log_L(|R|)$ steps, we will only have tuples in processors with $k_2 = 1$. Each processor $(k_1, 1)$ should contain $O(1)$ distinct tuples and their counts. In each of these steps, the number of tuples received by a processor $p$ would correspond to the number of distinct values of $(x_1, x_2,\ldots, x_{K'})$ that hash to the same value in $\left\lbrace 1, 2,\ldots, |R| \right\rbrace$, times $L$ (the number of processors sending tuples to $p$), which is $O(L)$ (up to log factors as specified earlier). At this stage, for each value $(x_1, x_2,\ldots, x_{K'}) \in (X_1, X_2,\ldots, X_{K'})$, we have its total count, which equals its degree in $R$, as needed. 

In the sequential setting, we can simply have one processor simulate the MapReduce computation above. Its computation cost equals the sum of computation and communication costs of all Mappers and Reducers in all rounds. The total computation is fully subsumed by the total communication of the MapReduce algorithm, which is $O(|R|)$.
\end{proof}

\section{Recovering previous results using DARTS}\label{sec:previous-recovery}

\subsection{Proof of Proposition~\ref{prop:nprr}}
\begin{proposition*}
If we compute the join using a single Light transform, our total cost is $\leq$ the AGM bound, thus recovering the result of the NPRR algorithm~\cite{Ngo:2012:WOJ:2213556.2213565}.
\end{proposition*}
\begin{proof}
If we perform a light transform, with set $X$ equal to the set of all attributes in the join, then $\DBP(G,X)$ simply equals the DBP bound on the join. Theorem~\ref{table:dbp-vs-agm} tells us that this is less than the AGM bound on the join. Moreover, after the light transform, the resulting join only has a single relation $R_X$ whose size equals the DBP bound on the join. Hence the $P$ and $Q$ values of the original join equal its DBP bound, and are $\leq$ the AGM bound.
\end{proof}

\subsection{Proof of Proposition~\ref{prop:yannakakis}}
\begin{proposition*}
If we successively apply the Split transform on an $\alpha$-acyclic join, with $G_1$ being an ear of the join in each step, then the total cost of our algorithm becomes $O(\IN + \OUT)$, recovering the result of Yannakakis' algorithm~\cite{Yannakakis81}.
\end{proposition*}
\begin{proof}
We proce that the $Q$ of an $\alpha$-acyclic join is $O(\IN)$, which implies the proposition.
We use induction on the number of relations in the join. It is clearly true when we have only $1$ relation. Suppose $Q$ equals input size for $\alpha$-acyclic joins with $\leq n-1$ relations, and consider an $\alpha$-acyclic join with $n$ relations. Because it is $\alpha$-acyclic, it has an `ear' i.e. it has a relation $R_1$ and a relation $R_2$ such that each attribute on $R_1$ is either unique to it, or is an attribute of $R_2$ as well. We apply the Split transform with $S = \attr(R_1) \cap \attr(R_2)$. Since this is a join, $\mO$ consists of all attributes, hence $S \subseteq \mO$. This lets us use the bound:
$$Q(G) \leq P(G'_1) + Q(G''_1) + Q(G''_2)$$
$G'_1$ has only one relation ($R_1$), so $P(G'_1)$ is $O(\IN)$. Similarly, $Q(G''_1)$ is $O(\IN)$. Consider $G''_2$, which consists of a relation $R_S$ and the relations in the original join other than $R_1$. The attributes of $R_S$ are a subset of the attributes of $R_1$. We do a light transform with $X = \attr(R_2)$. $\mR_X$ definitely includes $R_2$ and $R_S$. Since the attributes in $X$ are all contained in $R_2$, the DBP bound on this join is at most the size of $R_2$ i.e. $O(\IN)$. Moreover, the resulting join after the light transform has at most $n-1$ relations and is $\alpha$-acyclic. By the inductive hypothesis, its $Q$ value is $\IN$. Thus the $Q$ value of the whole join is at most $O(\IN) + O(\IN) + O(\IN) + O(\IN) = O(\IN)$, which completes the proof.
\end{proof}

\subsection{Proof of Proposition~\ref{prop:fhw}}
\begin{proposition*}
If a query has fractional hypertree width equal to fhw, then using a combination of Split and Light transforms, we can bound the cost of running DARTS by $O(\IN^{fhw} + \OUT)$, recovering the fractional hypertree width result.
\end{proposition*}
\begin{proof}
If fhw is the fractional hypertree width of the join, it means there exists a GHD~\cite{gottlob:ght,chekuri:conjunctive} such that the highest value of the AGM bound on the bags of the GHD equals $\IN^{fhw}$. For each bag $B$ of the GHD, we perform a Light transform with $X$ equal to the set of attributes in the bag. The time taken for computing $R_X$ is then the DBP bound on that join, which is less than the AGM bound, which is $\leq \IN^{fhw}$ (by the way the GHD was chosen). After all these light transforms, we are left with an $\alpha$-acyclic join, where each relation size is $\leq \IN^{fhw}$. Using Proposition~\ref{prop:yannakakis}, DARTS can process this join in time $\IN^{fhw} + \OUT$, proving that DARTS recovers the fhw bound.   
\end{proof}

\subsection{Proof of Proposition~\ref{prop:ayz}}
\begin{proposition*}
A cycle join of length $n$ with all relations having size $N$, can be processed by DARTS in time $O(N^{2 - \frac{1}{1 + \lceil \frac{n}{2}\rceil}} + \OUT)$, recovering the result of the AYZ algorithm~\cite{Alon:1994:FCG:647904.739463}. 
\end{proposition*}
\begin{proof}
Let $R_i$ be the relation with schema $A_i, A_{i+1}$ ($R_n$ has schema $A_n, A_1$).
Our proof follows the AYZ algorithm described in Section~\ref{sec:AYZ}. Let $\Delta = N^{\frac{1}{1 + \lceil \frac{n}{2} \rceil}}$. Then in degree configuration where at least one attribute $A_1$ has degree $> \Delta$ in a relation, we perform a heavy transform on $A_i$. The number of distinct $A_i$ values is at most $\frac{N}{\Delta}$. 
Thus $Q(G) \leq \frac{N}{\Delta}Q(G')$. Since $G'$ is $\alpha$-acyclic, $Q(G') \leq N$. Thus $Q(G) \leq \frac{N^2}{\Delta} = N^{2 - \frac{1}{1 + \lceil \frac{n}{2}\rceil}}$. Now consider degree configurations where all attributes $A_i$ have all degrees $\leq \Delta$. Then we perform a sequence of $n-2$ light transforms. In the $(2i+1)^{th}$ step, we perform a light transform with $X = \{A_1, A_2,\ldots,A_{i+3}\}$. And in the $(2i+2)^{nd}$ step, we perform a light transform with $X = \{A_1, A_n, A_{n-1} \ldots A_{n-i-1}\}$. The DBP bound for the $R_X$ in the $(2i+1)^{st}$ and $(2i+2)^{nd}$ transform is $\leq N \Delta^{i+1}$. This can be proved inductively. For $i=0$, setting cover $C = \{(R_1, \{_1,A_2\}), (R_1, \{A_2, A_3\})\}$. The solution to the linear program has $w_{R_1, \{A_1,A_2\}} = w_{R_2, \{A_3\}} = 1$ and other values $0$, which gives a output size bound of $N\Delta$. The $\{A_1,A_n,A_{n-1}\}$ case is similar. Now assume the inductive hypothesis for upto $i-1$. For $i$, We consider a cover $C = \{(R'_i, \attr(R'_i)), (R_{i+2}, \{A_{i+2}, _{i+3}\})\}$ where $R'_i$ is the relation with schema $\{A_1,A_2\ldots,A_{i+2}\}$ that was obtained from the last to last light transform. $|R'_i| = N\Delta^{i}$ by the inductive hypothesis. Then the solution to the linear program is $w_{R'_i, \attr(R'_i)} = w_{R_{i+2}, \{A_{i+2}, A_{i+3}\}} = 1$, giving the required bound of $N\Delta^{i+1}$. THe other case (for $(2i+2)^{nd}$ light transform is similar). At the end of these transforms, we will have two relations, $R_l$ with schema $A_1,A_2,\ldots,A_{\lceil \frac{n}{2}\rceil}$ and size $N\Delta^{\lceil\frac{n}{2}\rceil-1}$, and relation $R_r$ with schema $A_1,A_n,\ldots,A_{\lceil \frac{n}{2}\rceil}$ and size $\leq N\Delta^{\lceil\frac{n}{2}\rceil-1}$. Since these two relations now form an $\alpha$-acyclic join (any two relations form an $\alpha$-acyclic join), we use Proposition~\ref{prop:yannakakis} to join them in time $O(N^{2 - \frac{1}{1 + \lceil \frac{n}{2}\rceil}} + \OUT)$ as required. 
\end{proof}

\section{Subquadratic Joins}\label{sec:1-series-parallel-proofs}

\subsection{Tree-Cycle Structures}\label{sec:TCS}
We mention s simple extension of the AYZ result.
\begin{definition}
{\em Tree-Cycle Structure (TCS)}: 
\begin{enumerate}
\item A cycle of any length (including 1, which gives a single isolated node) is a TCS
\item If $T_1$ and $T_2$ are two disjoint TCSs, then adding an edge from any vertex of $T_1$ to any vertex of $T_2$ gives a new TCS. 
\item All TCSs can be formed by the above two steps. 
\end{enumerate}
\end{definition}

We can show that joins on TCSs can be processed in subquadratic time as well.
\begin{theorem}
\label{theorem:TCS-join}
A join over a TCS $T= (V, E)$ can be found in time $O(N^{2 - \frac{1}{1 + \lceil \frac{n}{2} \rceil}} + \OUT)$, where $n$ is the length of the longest cycle in the TCS.
\end{theorem}

The definition of Tree-Cycle structures can be extended to include other graphs for which we show subquadratic solvability.

\subsection{Subquadratic $1$-series-parallel graphs}
\begin{lemma}\label{lemma:s-t-edge}
If we have a $1$-series-parallel graph, which has a direct edge from $X_S$ to $X_T$ (i.e. a path of length $1$), then a join on that graph can be processed in subquadratic time.
\end{lemma}
\begin{proof}
Let $\mZ$ be the set of paths of length $> 1$ from $X_S$ to $X_T$. We use induction on $|\mZ|$. If $|\mZ| = 0$, then the join is just a single edge, which gets processed in time $O(N)$. Now assume we have a subquadratic ($N^{2 - \epsilon_{k-1}}$) solution for $|\mZ| = k$, and let $|\mZ| = k$. Now for any $Z \in \mZ$, we perform a split transform, with articulation set consisting $S = \left\lbrace X_S, X_T \right\rbrace$, and $G_1$ consisting of the attributes of $Z$. Since $G_1$ is now a cycle, it's $Q$ is $\leq N^{2-\epsilon}$ for some $\epsilon > 0$. And since $S \subseteq \mO$, we have 
\begin{align*}
Q(G) &\leq P(G'_1) + Q(G''_1) + Q(G''_2) \\
&= P(G'_1) + N^{2-\epsilon} + N^{2 - \epsilon_{k-1}} \\
\end{align*}
So to show subquadraticness, it suffices to show that $P(G'_1)$ is subquadratic. To do this, suppose the length of path $Z$ is $n$. Let $\delta = N^{\frac{1}{N + 2}}$. 

\squishlist
\item Suppose all attributes in $G_1$ have degree $\leq \delta$. Then we perform a sequence of light transforms until the join is solved, at a total cost of $N\delta^{n}$ which is subquadratic.
\item Suppose the path Z is given by $X_0 = X_S, X_1, \ldots, X_n = X_T$. If any attribute $X_l$ in $G_1$ has degree $> \delta$, we perform a heavy transform on it. After a heavy transform, we are left with a chain $X_{l+1}, X_{l+2}, \ldots, X_T, X_S, X_1, \ldots, X_{l-1}$. Then we perform a split tranform with articulation set $X_{l-2}$, and $G_1$ consisting of $X_{l-2}, X_{l-1}$. Since the output attribute set consists of $X_S, X_T$, which lies entirely in $G_2$, we use the split bound 
$$P(G) \leq P(G'_1) + P(G_2)$$
Here, the $P(G'_1)$term is simply $N$, so this split transforms effectively removes $X_{l-1}$ from the chain. We can similarly remove remaining attributes from the edges,leaving only $X_S$ and $X_T$, which gives a $P$ value of $N$, which is subquadratic.
\squishend

This shows that the join can be processed in subquadratic time. 
\end{proof}

\begin{lemma}\label{lemma:s-t-2-path}
Suppose we have a $1$-series-parallel graph $G$, which does not have a direct edge from $X_S$ to $X_T$, but which has a vertex $X_U$ such that there is an edge from $X_S$ to $X_U$ and from $X_U$ to $X_T$ (i.e. a path of length $2$ from $X_S$ to $X_T$). Let $G'$ be the graph obtained by deleting the vertex $X_U$ and edges $X_SX_U$ and $X_UX_T$. Then the join on $G$ can be processed in subquadratic time if and only if that on $G'$ can be processed in subquadratic time.
\end{lemma}
\begin{proof}
One direction of the lemma is easy to prove: If $G'$ requires quadratic time to solve, then by setting $X_SX_U$ and $X_UX_T$ to be full Cartesian products, we make join $G$ equivalent to $G'$, which means it must take quadratic time.

Now assume $G'$ can be solved in time $N^{2-\epsilon}$ for some $\epsilon > 0$. Firstly, if $X_U$ has degree $> N^{1 - \frac{\epsilon}{2}}$ in either relation, then we perform a heavy transform on $X_U$, giving a total cost of $\leq N^{2- \frac{\epsilon}{2}}$, which is subquadratic. So now assume the degree of $X_U$ is $\leq N^{1 - \frac{\epsilon}{2}}$. Then perform a light transform on $\{X_S, X_U, X_T\}$, to get a relation of size $\leq N^{2 - \frac{\epsilon}{2}}$. Then split with $G_1$ consisting of $X_S, X_U, X_T$. This gives a relation with attributes $X_SX_T$ of size $\leq N^{2 - \frac{\epsilon}{2}}$, to be added to $G_2$.

Now the proof is similar to the proof for the previous lemma. We again have an edge from $X_S$ to $X_T$, along with a number of other paths. Only this time, the edge relation has size $\leq N^{2 - \frac{\epsilon}{2}}$, rather than $= N$. But like before, we can choose a path $Z$, and let its length be $n$. Then we perform a split with articulation points $X_S, X_T$, and $G_1$ consisting of attributes of $Z$. Then we are left with a $P(G'_1)$, where $\mO_{G'_1} = \left\lbrace X_S, X_T \right\rbrace$. Like before, we choose a small enough $\delta$ ($ = N^{\frac{\epsilon}{2n+4}}$) such that if all attributes in $Z$ have degree $\leq \delta$ in relations if size $N$, then we perform a sequence of light transforms that give total cost $N^{2 - \frac{\epsilon}{2} + \delta^n}$ which is subquadratic. 

If the attributes don't all have degree $\leq \delta$ in relations of size $N$, then choose the smallest $l$ such that $X_l$ has degree $> \delta$ (where $Z$ is again written as $X_0 = X_S, X_1, \ldots, X_{n-1}, X_n = X_T$). Suppose its degree is $d$. Then we perform light transforms for $\{X_0, X_1, X_2\}$, $\{X_0, X_1,\ldots,X_3\}$, $\ldots$ $\{X_0,X_1\ldots,X_{l-1}\}$, which give a total cost of $N^{2 - \frac{\epsilon}{2} + (l-1)\delta}$, getting a relation $R_l$ with attributes $X_n, X_0, X_1,\ldots, X_{l}$. Let the degree of $X_l$ in $R_l$ be $d'$. Now we perform the heavy transform on $X_l$, which has at most $\min(\frac{N}{d}, \frac{|R_l|}{d'})$ distinct values. For each value, we get a chain, where each relation is of size $\leq N$, except for $R_l$ which has size $d'$. Then using split transforms like in the previous lemma proof, we can take out $X_{l+1}, X_{l+2}$ and so on one by one, and be left with $R_l$ alone, which is projected down to $\left\lbrace X_S, X_T \right\rbrace$. This gives a cost of $N + d'$ per $a \in X_l$. The total cost is thus $\min(\frac{N}{d}, \frac{|R_l|}{d'}) \times (N + d') \leq \frac{N^2}{\delta} + |R_l|$, which is subquadratic. This proves the lemma, as required.
\end{proof}

\subsection{$3$-SUM Hardness Proof}
We formally state and prove the lemma for $3$-SUM hardness of certain $1$-series-parallel graphs.
\begin{lemma}\label{lemma:s-t-3-paths}
Let $G$ be any $1$-series-parallel graph which does not have an edge from $X_S$ to $X_T$, but has $\geq 3$ paths of length at $\geq 3$ each, from $X_S$ to $X_T$. Then a join over $G$ can be processed in subquadratic time only if the $3$-SUM problem can be solved in subquadratic time.
\end{lemma}

We will reduce our join problem to the $3$-XOR problem. We only prove hardness for the simplest $1$-series-parallel graph having $3$ paths of size $= 3$ here. Joins on larger graphs can easily be reduced to this graph. Thus, we prove the theorem below (We use slightly different notation for the attribute names for convenience):

\begin{theorem*}
Consider a join over graph $G$ with attributes $A, B_1, C_1, B_2, C_2, B_3, C_3, D$, and relations $R_i(A, B_i), S_i(B_i,C_i), T_i(C_i, D) : \forall i \in \left\lbrace 1,2,3 \right \rbrace$, where each relation has size $N$. Suppose for some $c > 0$, there is an algorithm that processes the join in time $O(N^{2-c} + \OUT)$. Then $3$-SUM can be solved in time $O(N^{2 - t})$ for a $t > 0$. 
\end{theorem*}
\begin{proof}
We can assume that $N$ is a power of $2$. If it is not, we can simply introduce some dummy numbers while increase the problem size by at most a factor of $2$. Suppose we have a $c > 0$ and a corresponding algorithm. Now consider any $3$-XOR instance $x_1, x_2, ... x_N$. We will use the join algorithm to subquadratically solve this instance. We use a family of linear hash functions:

\stitle{Hash Function $h$:} For input length $l$ and output length $r$, the function $h$ uses $r$ $l$-bit keys $\bar{a} = (a_1, a_2, ... a_r)$ and is defined as $h_{\bar{a}}(x) = (\langle a_1, x \rangle, \langle a_2, x \rangle, ... \langle a_r, x \rangle)$ where $\langle a, b \rangle$ denotes inner product modulo $2$. 

This hash function is linear, i.e. $h(x) + h(y) = h(x+y)$ where addition is bitwise-xor. Also, $h_{\bar{a}}(0) = 0$ for all $\bar{a}$, and $\text{Pr}_{\bar{a}}\left[ h_{\bar{a}}(x) = h_{\bar{a}}(y)\right] \leq 2^{-r}$ for any $x \neq y$.

We pick a small $d > 0$ (the exact value will be specified later), and let $H = N^{1+d}$. Assume we picked the $d$ such that $N^{1+d}$ is a power of $4$ (We can always do this for sufficiently large $N$). We will hash down our numbers to $\left[ H \right]$, i.e. to $r = \text{log}(H)$ bits. The linearity of the hash function means that if $x_i + x+j  + x_k = 0$, then $h(x_i) + h(x_j) + h(x_k) = 0$ as well. On the other hand, if $x_i + x_j  + x_k \neq 0$, then the probability that $h(x_i) + h(x_j) + h(x_k) = 0$ is $\frac{1}{H}$. We will try to solve the $3$-XOR problem over the hashed values, and if the original problem has a solution ($3$ numbers that sum to $0$), then so will the hashed values. On the other hand, the expected number of false positives (triples of numbers that don't sum to zero, but whose hashed values sum to $0$) is given by the number of triples times the probability of a false positive, i.e. $\frac{N^3}{H} = N^{2-d}$.

Let $a = \frac{d}{4}$. We have $H$ buckets containing $N$ numbers total. Call a hash bucket {\em heavy} if it has more than $N^a$ elements. We would like to bound the number of elements that are contained in `heavy' buckets. 

We use a Lemma from Reference~\cite{BDP08}:
\begin{lemma}
Let $h$ be a random function $h : U \mapsto \left[H\right]$ such that for any $x \neq y$,
$\text{Pr}_{h}\left[h(x) = h(y)\right] \leq \frac{1}{H}$. Let $S$ be a set of $N$ elements, and let $B_h(x)$ $=$ $\{ y \in S \mid h(x) = h(y) \}$. For all $k$, we have
$$\text{Pr}_{h,x}\left[ |B_h(x)| \geq \frac{2N}{H} + k\right] \leq \frac{1}{k}$$
In particular, the expected number of elements from $S$ with $|B_h(x)| \geq \frac{2N}{H} + k$ is $\leq \frac{N}{k}$.
\end{lemma} 

Thus, the expected number of elements in `heavy' buckets is $N^{1-a}$, which is in $o(N)$. For each heavy element, we can try summing it with each other $x_i$, and see if the resulting sum is one of the $x_i$s. Thus, we can check the sum condition on all heavy elements in time $N^{2-a}$. Thus, we can now assume that all buckets have $< N^a$ elements. 

We now present an instance of the join that is reducible from the $3$-XOR problem instance. For each attribute $B_i$ and $C_i$, their values consist of all bit combinations with $\frac{(1+d)\text{log}(N)}{2}$ bits. Thus, there are $N^{\frac{1+d}{2}}$ distinct attribute values for each of those attributes. Attributes $A$ and $D$ have $N^{1+d}$ distinct attribute values each. Each relation $S_i(B_i, C_i)$ has up to $N$ edges as follows. For each $x_i$ from the original problem that was not in a heavy bucket, we express it's hash value as $h(x_i) = b_i + N^{\frac{1+d}{2}}c_i$. Then, we add an edge between values $b_i \in B_j$ and $c_i \in C_j$ for $j = 1, 2, 3$. For relations $R_j(A, B_j)$ and $T_j(C_j, D)$, we do the following: Consider all triples $t_{i,1}, t_{i,2}, t_{i,3}$ of $\frac{(1+d)\text{log}(N)}{2}$-bit numbers whose bitwise-xor is $0$. There are $N^{1+d}$ such triples. For each such triple, we take one element $a_i \in A$, and connect it to each of $t_{i,1} \in B_1$, $t_{i,2} \in B_2$, $t_{i,3} \in B_3$. Similarly, we take one element $d_i \in D$ and connect it to each of $t_{i,1} \in C_1$, $t_{i,2} \in C_2$, $t_{i,3} \in C_3$. Thus, we have a join instance with relations of size $N^{1+d}$. Setting up this join instance given the $3$-XOR instance takes time $O(N^{1+d})$. 

Now we analyze the output of this join instance. Suppose we have an output tuple $a \in A, b_1 \in B, b_2 \in B_2, b_3 \in B_3, c_1 \in C_1, c_2 \in C_2, c_3 \in C_3, d \in D$. From our relations, we know that there is an $x_{i,1}$ whose hash equals $b_1 + N^{\frac{1+d}{2}}c_1$, an $x_{i,2}$ whose hash equals $b_2 + N^{\frac{1+d}{2}}c_2$, and an $x_{i,3}$ whose hash equals $b_3 + N^{\frac{1+d}{2}}c_3$. Moreover, since $a$ is connected to $b_1, b_2, b_3$, we know that the bitwise xor $b_1 + b_2 + b_3 = 0$. Similarly, $c_1 + c_2 + c_3 = 0$. Hence the bitwise xor $h(x_{i,1}) + h(x_{i,2}) + h(x_{i,3}) = 0$. Thus, either the triple $(x_{i,1}, x_{i,2}, x_{i,3})$ is a solution to the $3$-XOR problem, or it is a false positive.

Now we apply the subquadratic join algorithm whose existence we assumed, on our join instance of size $N^{1+d}$. If it runs for time greater than $O(N^{(1+d)(2-c)} + N^{2-\frac{d}{2}})$, we terminate it and return `true' for the $3$-XOR problem (we will justify this later). Otherwise, for each output tuple, we get a triple of hash buckets whose bitwise xor is zero. For each such triple of buckets, we check the (at most $N^{3a}$) corresponding  triples of $x_i$'s and check if they sum to $0$. This takes time $O(N^{(1+d)(2-c)} + N^{2-d+3a}) = O(N^{(1+d)(2-c)} + N^{2-a})$. If we find such a triple, then we return true for the $3$-XOR problem. If we don't find such a triple for any of the outputs of the join, we return false. Now recall that the expected number of false positives is $N^{2-d}$. If the correct answer to the $3$-XOR problem is false, then the program should terminate in time $O(N^{(1+d)(2-c)} + N^{2-d})$ with high probability. This justifies our decision to return true if the program runs for a polynomially longer time, as the probability of the correct answer falls exponentially as the program keeps running beyond $O(N^{(1+d)(2-c)} + N^{2-d})$. 

This means we can solve the $3$-XOR problem with high probability, in time $O(N^{2-a} + N^{1+d} + N^{(1+d)(2-c)} + N^{2-d})$. So we choose $d$ small enough such that $(1+d)(2-c) < 2$, and set $t = \min(a, 1-d, 2 - (1+d)(2-c))$. This way, $3$-XOR can be solved in time $N^{2-t}$, proving the theorem. 
\end{proof}

\subsection{DARTS Application examples}\label{sec:darts-examples}
\begin{example}
Joins over $K_{2,n}$, a special case of $1$-series parallel graphs, have some potential applications for recommendations. $K_{2,n}$ consists of attributes $X, Z$ on one side, connected to each of $Y_1, Y_2,\ldots, Y_n$ on the other side. Joining over $K_{2,n}$ where each relation is an instance of a friendship graph gives us pairs of people who have at least $n$ friends in common, along with the list of those friends. If instead the $X$ attribute is a netflix user id, $Z$ is a movie id, and $Y_i$s are attributes such as genres, then the join could be interpreted to mean ``find user-movie pairs such that the user likes at least $n$ attributes of the movie''.

As an example of using DARTS for $1$-series-parallel graphs, we prove that a join over $K_{2,n}$ can be processed in subquadratic time. The join has relations $R_i(X, Y_i)$, $S_i(Y_i, Z)$ for all $1 \leq i \leq n$. We prove that the $Q$ of the join is subquadratic using induction on $n$.

{\em Base Case:} If $n = 1$, the graph is a chain, and can be solved in linear time using Yannakakis' algorithm. Since DARTS includes Yannakakis' algorithm as a special case, it can solve the chain in linear time as well i.e. $Q = O(N)$.

{\em Induction:} Now we assume that $Q$ for $K_{2,n}$ is $\leq N^{2-\delta_n}$, for some $\delta_n > 0$. Consider $K_{2,n+1}$. For any degree configuration $c$ in which at least one of the $Y_{i}$'s has a degree greater than $N^{1 - \frac{\delta_n}{2}}$, we perform the heavy transform on that $Y_i$. The number of $Y_i$'s is less than $N^{\frac{\delta_n}{2}}$, and the reduced graph is a $K_{2,n}$, which has $Q \leq N^{2-\delta_n}$. Thus, the heavy transform gives us $Q \leq N^{2 - \frac{\delta_n}{2}}$ for configuration $c$ of $K_{2,n+1}$. On the other hand, if the degree configuration $c$ has all $Y_i$s having degree $\leq N^{1 - \frac{\delta_n}{2}}$, then we perform light transforms on $\{X, Y_i, Z\}$ for each $i$ one by one, and end up with relations $R_{Y_i}(X, Y_i, Z)$ of size $\leq N^{2 - \frac{\delta_n}{2}}$. Now for each $i$, we perform Split transforms using articulation set $\left \lbrace X,Z \right\rbrace$, and $G_1$ consisting of $X, Y_i, Z$. Then $P(G_1) \leq N^{2 - \frac{\delta_n}{2}}$ and the projection onto $XZ$ is of size $\leq N^{2 - \frac{\delta_n}{2}}$ as well. This upper bounds $Q$ by $N^{2 - \frac{\delta_n}{2}}$. Thus, the $Q$ of $K_{2,n+1}$ is subquadratic, which completes the induction.
\end{example}

\begin{example}
The runtime improvements of DARTS are not limited to treewidth $2$ joins. In general, we can process the join on the complete bipartite graph $K_{m,n}$, which has treewidth $\min(m,n)$, in time $O(\IN^{\min(m,n) - \epsilon_{m,n}} + \OUT)$, where $\epsilon_{m,n} > 0$. Simply marginalizing on an attribute on the $m$ side gives us a time bound of $O(\IN^{\min(m-1,n) - \epsilon_{m-1,n} + 1} + \OUT)$. But we get $\epsilon_{m,n} > \epsilon_{m-1,n}$, which means DARTS does more than just marginalize on an attribute. 

For example, consider $K_{3,3}$ which has treewidth $3$. All relations have size $N$, and let the attributes be $X_1$, $X_2$, $X_3$ on one side and $Y_1$, $Y_2$, $Y_3$ on the other. Suppose we can process a join over $K_{2,3}$ in time $O(N^{2 - \epsilon_{2,3}} + \OUT)$. Set $\Delta = N^{(2 - \epsilon_{2,3})/3}$. If the degree for an attribute is $> \Delta$, we could marginalize on it and achieve a runtime of $O(\Delta^{-1} N^{3 - \epsilon_{2,3}} + \OUT)$. On the other hand, if all degrees are $\leq \Delta$, then we can perform a Light transform on $\{X_1, X_2, X_3, Y_i\}$ for each $i$ to get $3$ relations of size $\leq N\Delta^2$ each. We can join them using Split transforms, getting a runtime of $O(N\Delta^2 + \OUT)$. Either way, the runtime of DARTS is bounded by $O(N^{3 - \epsilon_{3,3}} + \OUT)$ where $\epsilon_{3,3}$ $=$ $3 - (1 + 2((2 - \epsilon_{2,3})/3))$ $=$ $2(1 + \epsilon_{2,3})/3 < 1 + \epsilon_{2,3}$. 
\end{example}

\section{Proofs on $m$-width and $\MO$ bound (Section~\ref{sec:mo-bound},~\ref{sec:m-width})}\label{sec:mo-mw-proofs}

\subsection{Proof of Theorem~\ref{prop:mo-bound}}
We first state and prove a more general proposition, and Proposition~\ref{prop:mo-bound} will be a corollary.

\begin{proposition}\label{prop:mo-bound-general}
For all $A \subseteq \mA$, we can compute a relation $R_A$  in time $O(\IN^{m_A})$ such that (i) $|R_A| \leq \IN^{m_A}$ (ii) $\pi_{A}(\Join_{R \in \mR} R) \subseteq R_A$ (where $m_A$ is as defined in Section~\ref{sec:m-width}).
\end{proposition}
\begin{proof}
For each $A \subseteq \mA$, let $O_A = \pi_{A}(\Join_{R \in \mR} R)$. Fix any $A \subseteq \mA$ and consider the solution to $\mathsf{Prog}(A)$. In the solution, there must be at least one tight constraint of the form $s_A \leq s_B$ for $A \subseteq B$ or $s_A \leq s_B + d(P,Q,R)$ for some $P$, $Q$, $E$, $R$ such that $P \subseteq Q \subseteq \attr(R)$, $B = P \cup E$, $A = Q \cup E$. Then in turn, there must be a similar constraint on $s_B$. The only constraint in the system that does not have one relation on the LHS and one on the RHS is the $s_{\emptyset} = 0$ constraint. 

Thus there must be a chain $A_0,A_1,\ldots,A_k$ such that $A_0 = \emptyset$, $A_k = A$ and there is a tight constraint with $A_{i+1}$ on the LHS and $A_i$ on the RHS (i.e. $A_{i+1} \leq A_i + \ldots$. Then we produce a sequence of relations $R_0,\ldots,R_k$ such that for all $i$ : $|R_i| \leq \IN^{m_{A_i}}$. The final $R_k$ equals our $R_A$. We produce these relations inductively: If $A_{i+1} \subseteq A_i$, then we set $R_{i+1} = \pi_{A_{i+1}} R_i$. Otherwise, there exist $P$, $Q$, $R$, $E$ such that $P \subseteq Q \subseteq \attr(R)$, $A_i = P \cup E$, $A_{i+1} = Q \cup E$ and $s_{A_{i+1}} = s_{A_i} + d(P,Q,R)$. Then we set $R_{i+1} = R_i \Join \pi_{Q}(R)$. Since these operations only involve relations in the original join, all $R_i$s satisfy $O_{A_i} \subseteq R_i$. Moreover, for all $i$, $|R_i| \leq \IN^{s_{A_i}}$. Thus, $R_k$ is computed in time $O(\IN^{s_{A_k}}) = O(\IN^{m_A})$, and satisfies $O_A \subseteq R_k$. Setting $R_A = R_k$ gives us the required $R_A$ satisfies conditions (i) and (ii) of the proposition, completing our proof. 
\end{proof}

\begin{proposition*}
The output size $\Join_{R \in \mR} R$ is in $O(\IN^{m_{\mA}})$.
\end{proposition*}
\begin{proof}
For each $A \subseteq \mA$, let $O_A = \pi_{A}(\Join_{R \in \mR} R)$. We set $A = \mA$ in Proposition~\ref{prop:mo-bound-general}. $O_{\mA}$ is simply the output of the join $\Join_{R \in \mR}R$ and since it is a subset of $R_{\mA}$ which has size $\leq \IN^{m_{\mA}}$, the output itself must have size $O(\IN^{m_{\mA}})$.
\end{proof}

\subsection{Proof of Theorem~\ref{thm:mw-join}}
\begin{theorem*}
Any join query can be answered in time $O(\IN^{\MW} + \OUT)$, where $\MW$ is its $m$-width.
\end{theorem*}
\begin{proof}
For all $A \subseteq \mA$, let $O_A$ denote $\pi_{A} (\Join_{R \in \mR} R)$. Given a GHD $(\mT, \chi)$ with $m$-width equal to $\MW$, we perform the join in three steps:
\begin{itemize}
\item For each bag $\chi(t)$ of the GHD, we compute $R_{\chi(t)}$ like in Proposition~\ref{prop:mo-bound-general}. That is, we compute $R_{\chi(t)}$  in time $O(\IN^{m_{\chi(t)}})$ such that (i) $|R_{\chi(t)}| \leq \IN^{m_{\chi(t)}}$ (ii) $O_{\chi(t)} \subseteq R_{\chi(t)}$. The latter property ensures that $O_{\mA} \subseteq \Join_{t \in \mT} R_{\chi(t)}$. Moreover, by definition of $m$-width, the computation time for each $R_{\chi(t)}$ and the size of $R_{\chi(t)}$ are bounded by $O(\IN^{\MW})$.
\item Then for each bag $\chi(t)$, we compute $R'_{\chi(t)}$ which is $R_{\chi(t)}$ semi-joined with $\pi_{{\chi(t)}}(R)$ for each $R \in \mR$. This ensures that $O_{\mA} = \Join_{t \in \mT} R_{\chi(t)}$. Moreover, $|R'_{\chi(t)}| \leq |R_{\chi(t)}| \leq \IN^{\MW}$.
\item  Then we use Yannakakis' algorithm to join all the $R_{\chi(t)}$'s. This can be done in time $O(\IN^{\MW} + \OUT)$, completing the proof.
\end{itemize}
\end{proof}

\subsection{Proof of Theorem~\ref{thm:mo-bound}}
\begin{theorem*}
For any join query $\mR$, and any degree configuration $c \in \mC_2$, $\MO(\mR(c)) \leq \DBP(\mR(c),2) + |C|\log(2)$, where $C$ is the cover used in the DBP bound.
\end{theorem*}
\begin{proof}
$\DBP(\mR(c),2)$ is obtained by solving Linear Program~\ref{lp:degree-packing-primal} for the optimal cover. Let $C$ be the optimal cover, and $v_a$ be the value in the optimal solution for each $a \in \mA$. And for each $A \subseteq \mA$, let $s_A$ denote the value in the optimal solution for the linear program $\textsf{Prog}(\mA)$.

Let $C = \{(R_1,A_1), (R_2,A_2),\ldots,(R_{|C|},A_{|C|})\}$, where $R_i \in \mR$ and $A_i \subseteq \attr(R_i)$ for all $i$. Define $B_j = \bigcup_{i=1}^{j}A_i$ for all $1 \leq j \leq |C|$. Since $C$ is a cover, we must have $B_{|C|} = \mA$. 

Now for each $j$, we will show that $s_{B_j} \leq j\log(2) + \sum_{a \in B_j} v_a$. We do this using induction on $j$. Then for $j=|C|$ the LHS $s_{B_{|C|}}$ equals $\MO(\mR(c))$ and RHS $|C|\log(2) + \sum_{a \in \mA} v_a$ equals $\DBP(\mR(c),2) + |C|\log(2)$, proving our theorem.

Base Case: For $j=1$, setting $R = R_1$, $A = A_1$, $A' = A_1$ for Linear Program~\ref{lp:degree-packing-primal} gives us the constraint $\sum_{a \in A_1} v_a \geq \log(d_{\pi_{A}(R),\emptyset}/2)$. And $\textsf{Prog}(\mA)$ with $A = \emptyset$, $B = A_1$, $R = R_1$, $E = \emptyset$ gives us the constraint $s_{A_1}$ $\leq s_{\emptyset} + d(\emptyset, A_1, R_1)$ $= \log(d_{\pi_{A}(R),\emptyset})$ $\leq \log(2) + \sum_{a \in A_1} v_a$. Then since $B_1 = A_1$, our base case is proved.

Induction: Suppose we have proved $s_{B_j} \leq j\log(2) + \sum_{a \in B_j} v_a$ for $j-1$. Now let $E_j = B_j \setminus B_{j-1}$. Then Linear Program~\ref{lp:degree-packing-primal} with $R = R_j$, $A = A_j$, $A' = E_j$ gives us $\sum_{a \in E_j} v_a \geq \log(d_{\pi_{A_j}(R_j),A_j \setminus E_j}/2)$. $\textsf{Prog}(\mA)$ with $R = R_j$, $A = A_j \setminus E_j$, $B = A_j$, $E = B_{j-1}$ gives us $s_{B_{j-1} \cup A_j}$ $\leq$ $s_{(A_j \setminus E_j) \cup B_{j-1}}$ $+$ $d_{A_j \setminus E_j, A_j, R_j}$. Now $B_j = B_{j-1} \cup A_j$ by definition of $B_j$, and $(A_j \setminus E_j)$ $\cup$ $B_{j-1}$ $=$ $B_{j-1}$ since $A_j \subseteq B_j$ $=$ $E_j \cup B_{j-1}$. So $s_{B_j}$ $\leq s_{B_{j-1}} + \log(d_{\pi_{A_j}(R_j),A_j \setminus E_j})$ $\leq s_{B_{j-1}} + \log(2) + \sum_{a \in E_j} v_a$. And by inductive hypothesis, $s_{B_{j-1}} \leq (j-1)\log(2) + \sum_{a \in B_{j-1}} v_a$. This gives us $s_{B_j} \leq j\log(2) + \sum_{a \in B_j} v_j$.

This proves that $s_{B_j} \leq j\log(2) + \sum_{a \in B_j} v_a$ for all $j$, and consequently that $\MO(\mR(c)) \leq \DBP(\mR(c),2) + |C|\log(2)$, completing our proof.
\end{proof}

\subsection{Recovering DARTs results using GHDs}\label{sec:ghd-darts-recovery}
Theorem~\ref{thm:mo-bound} shows that the $\MO$ bound is smaller than the AGM bound. As a result, the $\MW$ of a GHD is smaller than its fhw. This lets us recover Propositions~\ref{prop:nprr}-\ref{prop:fhw}. We now show how to recover the subquadratic join results from Theorem~\ref{thm:1-series-parallel-graphs} and the AYZ result.

\subsubsection{Recovering AYZ}
\begin{proposition*}
A cycle join of length $n$ with all relations having size $N$, has $m$-width $\leq 2 - \frac{1}{1 + \lceil \frac{n}{2}\rceil}$, recovering the result of the AYZ algorithm. 
\end{proposition*}
The cycle join has relations $R_1(X_1,X_2),\ldots,R_n(X_n,X_1)$ of size $N$ each. Choose $\Delta = N^{\frac{1}{1 + \lceil n/2 \rceil}}$ as before. We will show that for each degree configuration, we can construct a GHD that has $\MW \leq 2 - \frac{1}{1 + \lceil n/2 \rceil}$.

Suppose the configuration is such that the degree of some $X_k$ if $\geq \Delta$, then we build a GHD with a bags $\{X_k\} \cup \attr(R_j)$ for each $j$. The bags form a chain $\{X_k\} \cup R_{k}$, $\{X_k\} \cup R_{k+1}$, $\{X_k\} \cup R_{k+2}$, $\ldots$, $\{X_k\} \cup R_{k-2}$, $\{X_k\} \cup R_{k-1}$, which gives us the GHD. The $m$ value for each bag is bounded by $\log(N^2\Delta^{-1})$ since $m_{\attr(R_j)} \leq \log(N)$ and using $A = \emptyset$, $B = \{X_k\}$, and $d(A,B,R_k) \leq \log(N\Delta^{-1})$.

If all degrees in the configuration are $\leq \Delta$, we form a GHD with two bags: $\{X_1$, $X_2,\ldots$, $X_{\lceil n/2 \rceil}\}$ and $\{X_1$, $X_n$, $X_{n-1},\ldots$, $X_{\lceil n/2 \rceil}\}$. The $m$ value of each bag is still $N\Delta^{\lceil n/2 \rceil} = N^2\Delta^{-1}$. This time, we have $m_{\{X_1,X_2\}} \leq \log(N)$ and for each $i$, $m_{\{X_1,\ldots,X_{i+1}\}} \leq m_{\{X_1,\ldots,X_i\}} + \log(\Delta)$ since $d(A,B,R) = \log(\Delta)$ for $A = \{X_i\}$, $B = \{X_i,X_{i+1}\}$, $R = R_i$. 

Thus for each degree configuration, we can find a GHD with $\MW \leq N^2\Delta^{-1}$, which implies that $m$-width is $\leq 2 - \frac{1}{1 + \lceil n/2 \rceil}$, which lets us recover the AYZ result.

\subsubsection{Lemma~\ref{lemma:s-t-edge}}
\begin{lemma*}
If we have a $1$-series-parallel graph, which has a direct edge from $X_S$ to $X_T$ (i.e. a path of length $1$), then the $m$-width of a join over the graph is $< 2$.
\end{lemma*}
\begin{proof}
Once again, we will show that for any degree configuration, we can construct a GHD with $\MW < 2$. Suppose there are $k$ paths from $X_S$ to $X_T$ excluding the $X_SX_T$ edge. Each of the $k$ paths, along with edge $X_SX_T$ forms a cycle. For each cycle, we form a GHD for the given degree configuration like we did for the AYZ recovery. Call these GHDs $D_1,D_2,\ldots,D_k$. Since we have an edge $X_SX_T$, each $D_i$ contains at least one bag $B_i$ that contains both $X_S$ and $X_T$. We create a new bag $\{X_S,X_T\}$, and connect it to each $B_i$ for $1 \leq i \leq k$. This gives us a GHD for the full join, and the $m$ value of its bags is no more than it was in the original GHDs, which was shown to be $< 2$ when we recovered AYZ. As a result, when there is a $X_SX_T$ edge, we have GHD with $\MW < 2$ for every degree configuration, and thus the $m$-width of the join is $< 2$.
\end{proof}

\subsubsection{Lemma~\ref{lemma:s-t-2-path}}
\begin{lemma*}
Suppose we have a $1$-series-parallel graph $G$, which does not have a direct edge from $X_S$ to $X_T$, but has a vertex $X_U$ such that there is an edge from $X_S$ to $X_U$ and from $X_U$ to $X_T$ (i.e. a path of length $2$ from $X_S$ to $X_T$). Let $G'$ be the graph obtained by deleting the vertex $X_U$ and edges $X_SX_U$ and $X_UX_T$. Then the $m$-width of a join on $G$ is $< 2$ if and only if the $m$-width of the join on $G'$ is $< 2$.
\end{lemma*}
\begin{proof}
We have edge $X_SX_U$ and $X_UX_T$ and no direct edge $X_SX_T$. As before, one direction is easy to prove. Suppose the $m$-width of the join over $G$ is $< 2$. That is, the join on $G$ has a GHD with $\MW < 2$ for all degree configurations. Then for any configuration $c'$ for $G'$, consider the corresponding configuration $c$ for $G$ where $X_U$ has degree $N$ in both its relations and other degrees are the same. Consider the GHD with $\MW < 2$ for this configuration on $G$. We have $s_{\{X_U\}} = 0$ and $s_{A} = s_{A \cup \{X_U\}}$ for all $A \subseteq \mA$. Then the GHD obtained by removing $X_U$ from each bag gives us a GHD for $G'$ with $\MW < 2$. This implies that the $m$-width of the join over $G'$ is also $< 2$.

Now suppose the $m$-width of the join over $G'$ is $< 2$. That is, there is an $\epsilon$ such that for each degree configuration for $G'$, there is a GHD with $\MW \leq 2-\epsilon$. Now consider any degree configuration $c$ for $G$ and the configuration $c'$ for $G'$ obtained by keeping the same degrees for all values (not in $X_U$). Suppose $X_U$ has degree $\geq N^{1 - \frac{\epsilon}{2}}$, then $s_{\{X_U\}} \leq \epsilon/2$. Let $D'$ be a GHD of $G'$ with $\MW < 2 -\epsilon$. Adding $X_U$ to each bag of GHD $D'$ gives us a GHD for $G$ that has $\MW < 2 - \epsilon/2$. 

So now we can assume that the degree of $X_U$ is $\leq N^{1 - \frac{\epsilon}{2}}$ in both its relations. Thus $s_{\{X_S,X_U,X_T\}} \leq 2 - \frac{\epsilon}{2}$. Now like in the previous proof, we will consider every other path from $X_S$ to $X_T$, and construct a GHD with $\MW < 2$ for each path, which has at least one bag containing both $X_S$ and $X_T$. Then we can create a new bag $\{X_S,X_T\}$ and use it to stitch all the GHDs together to get a GHD for $G$ that has $\MW < 2$. We now describe how to construct the $\MW < 2$ GHD for each path.

Consider any other path $X_1,X_2,\ldots,X_n$ where $X_1 = X_S$, $X_n = X_T$. Let our relations in the path be $R_1(X_1,X_2),\ldots,R_{n-1}(X_{n-1}X_n)$. Let $\delta = N^{\epsilon/(2n+4)}$. 
Suppose some $X_i$ has degree $\geq \delta$ in relation $R_i$. Choose the smallest such $i$, (so for all $j < i$, the degree of $X_j$ in $R_j$ is $\leq \delta$). Then we form a GHD with one bag $\{X_n,X_1,X_2,\ldots,X_i\}$, and also a bag $\{X_i\} \cup R_j$ for each $j > i$. The $m$ of the first bag is $\log(N^{2 - \frac{\epsilon}{2}}\delta^i)$ (because $m_{\{X_n,X_1\}} \leq \log(N^{2 - \frac{\epsilon}{2}})$ and each of $X_2,\ldots,X_i$ adds $\log(\delta)$ to it). From the definition of $\delta$, we have $N^{2 - \frac{\epsilon}{2}}\delta^i$ $\leq$ $N^2\delta^{-1}$. The $m$ of other bags is $\log(N^2\delta^{-1})$, since $m_{\{X_j,X_{j+1}\}} \leq \log(N)$ and $X_i$ adds at most $\log(N\delta^{-1})$. Thus the $\MW$ of the path GHD is $\leq 2 - \log(\delta)$. On the other hand, if no $X_i$ has degree $\geq \delta$ in any $R_i$, then a single bag $\{X_1,X_2,\ldots,X_n\}$ has $m$ $\leq \log(N^{2 - \frac{\epsilon}{2}}\delta^{n-2})$, which gives us a GHD for the path with $\MW < 2$. 

Thus for each degree configuration of $G$, we can construct a GHD with $\MW < 2$, which implies that the $m$-width of the join over $G$ is $< 2$.
\end{proof}

\subsection{Comparison to other widths}\label{sec:widths-comparison}
Theorem~\ref{thm:mo-bound} implies that $m$-width is no larger than fractional hypertreewidth (and consequently, no larger than treewidth and generalized hypertreewidth). $m$-width can even be smaller than submodular width (which, ignoring $m$-width, is the tightest known notion of width for general joins), as shown in the Example below. 

\begin{example}\label{example:m-beats-submodular}
Consider a cycle join with $n$ relations, with each relation having size $N$ and all degrees being equal to $1$ in each relation. Then the $m$-width of the join is given by $1$ (because all the $d(A,B,R)$ values in Linear Program~\ref{lp:mo-bound} are $0$ for $A \neq \emptyset$ and $1$ if $A = \emptyset$). On the other hand, the submodular width of this join is $2 - \frac{1}{1 + \lceil \frac{n}{2}\rceil}$.

Similarly, if we consider a clique join with $n$ attributes (i.e. for each pair of attributes, there is a single relation with $N$ tuples), and all degrees are $1$ in each relation, then the $m$-width of the join is $1$, while the submodular width is $n/2$, which can be unboundedly larger.
\end{example}

The above examples rely on the fact that $m$-width takes actual degrees of the relations following degree-uniformization into account, while submodular width uses worst-case degrees. In addition, whenever $m_A$ happens to be a submodular function over $\mA$, $m$-width is guaranteed to be $\leq$ submodular width. Unfortunately, $m_A$ is not always submodular, as shown by the example below:
\begin{example}
Consider a join with relations $R(A,B)$, $S(B,C)$, $T(B)$, $U(C)$. Let $|R| = |S| = N$, $|T| = |U| = \sqrt{N}$. And let the degree of each $A$ value in each relation be $\sqrt{N}$ (so there are $\sqrt{N}$ distinct $A$ values), while the degrees of $B$ and $C$ values are $1$ (so there are $N$ distinct $B$, $C$ values in $R$, $S$ and $\sqrt{N}$ values in $T$, $U$. Now we compute the $m$ values for different sets.

Since there are $N$ $B$, $C$ values in relations $R$, $S$, but only $\sqrt{N}$ $B$, $C$ values in relations $T$ and $U$, we have $m_{\{B\}}$ $= m_{\{C\}}$ $= \log(\sqrt{N})$, and $m_{\{A\}}$ is $\log(\sqrt{N})$ as well. Now for $m_{\{A,B\}}$, we have $s_{\{A,B\}} \leq s_{\{B\}} + d(\{B\}, \{A,B\}, R)$. Since the degree of $B$ is $1$, $d(\{B\}, \{A,B\}, R)$ is $0$, which gives us $m_{\{A,B\}} = \log(\sqrt{N})$ as well. Similarly, $m_{\{A,C\}} = \log(\sqrt{N})$. Finally, we have $m_{\{B,C\}}$ $= m_{\{A,B,C\}} = \log(N)$. Thus we have $m_{\{A\}} + m_{\{A,B,C\}} = \log(N\sqrt{N})$, while $m_{\{A,B\}} + m_{\{A,C\}} = \log(N)$, which implies that $m$ is not submodular. 
\end{example}

The above example gets to the heart of why our degree uniformization is weaker than Marx's uniformization (while being less expensive). Our degrees are uniform within relations, but not necessarily in the final output. For example, each $A$ value has degree $\sqrt{N}$ in the relations, but because only $\sqrt{N}$ out of $N$ $B$ and $C$ values will be in the output, the degree of an $A$ value in the output can range anywhere from $1$ to $\sqrt{N}$. Marx's uniformization ensures that degrees are uniform in certain projections of the output as well.

Even though we started with $\sqrt{N}$ values of $A$ each having degree $\sqrt{N}$, once most of the $B$ and $C$ values are eliminated due to relations $T$, $U$, both the number of matching $A$ values and their degrees are reduced. The number of $A$ values that still have degree $\sqrt{N}$ can now be at most $1$ (since there are $\sqrt{N}$ values of $B$, $C$ left). This change in the number of values is not taken into account in our $s$ values. One naive way to remedy this is to repeatedly perform degree-uniformization after every step of the join, but this can lead to a higher than linear cost.

\subsection{Relating subquadratic solvability to widths}\label{sec:submodular-width-lower-bound}
Each graph that we showed to be subquadratically solvable has $m$-width $< 2$ (and also submodular width $< 2$). Moreover, the $3$-SUM hard 1-series-parallel graph from Theorem~\ref{thm:1-series-parallel-graphs} can be shown to have $m$-width and submodular width equal to $2$. We show this next.

The graph has edges $X_SX_{A_1}$, $X_{A_1}X_{B_1}$, $X_{B_1}X_T$, $X_SX_{A_2}$, $X_{A_2}X_{B_2}$, $X_{B_2}X_T$, $X_SX_{A_3}$, $X_{A_3}X_{B_3}$, $X_{B_3}X_T$. Then we give a edge-dominated submodular function $f$ such that for any GHD, there must exist a bag $\chi(t)$ such that $f(\chi(t)) \geq 2$. Suppose there are $N$ values in $X_S$, $X_T$ with degree $1$ in each relation, and $\sqrt{N}$ values in other attributes with degree $\sqrt{N}$ in each relation. Then the $m$ values for this join happen to be submodular. Specifically, we have $m_{\{X_S\}}$ $= m_{\{X_T\}} = 1$, and for all $i$, we have $m_{\{X_S,A_i\}}$ $m_{\{X_T,B_i\}}$ $= 1$, $m_{\{A_i\}}$ $= m_{\{B_i\}}$ $= 0.5$, $m_{\{A_i,B_i\}} = 1$, $m_{\{X_S,B_i\}}$ $= m_{\{X_T,A_i\}}$ $= m_{\{X_S,A_i,B_i\}}$ $= m_{\{X_T,A_i,B_i\}}$ $= 1.5$, $m_{\{X_S,A_i,B_i,X_T\}}$ $= m_{\{X_S,B_i,X_T\}}$ $= m_{\{X_S,A_i,X_T\}}$ $= m_{\{X_S,X_T\}}$ $= 2$. Moreover, for all $i$, $j \neq i$, if $P_i = \{X_S,A_i,B_i,X_T\}$, $P_j = \{X_S,A_j,B_j,X_T\}$, and $P \subseteq P_i \cup P_j$ then $m_{P}$ $= m_{P \cap P_i}$ $+ m_{P \cap P_j}$ $- m_{P \cap P_i \cap P_j}$. $m_P$ for $P \subseteq P_i \cup P_j \cup P_k$ can be found similarly. 

Now any GHD that puts $X_S$ and $X_T$ together must have width $2$ since $m_{\{X_s,X_T\}} = 2$. But if $X_S$ and $X_T$ never occur together, then the path between their nodes in the GHD must contain each of the paths in the graph ($\{X_SA_i,A_iB_i,B_iX_T\}$ for all $i$). Thus each node in the path must contain at least one node from each path, and at least of them must contain the edge $A_1B_1$. This means that at least one node in the GHD must contain four of the $A_i$s and $B_i$s combined, which again makes the width $2$. This shows that the submodular width of the $3$-SUM hard graph is $2$. 

This may suggest that a join can be solved subquadratically if and only if its submodular width is $< 2$. However, this is not the case. In fact, submodular width is not the a tight lower bound on the runtime exponent. As a counterexample, a triangle join has submodular width equal to $3/2$. But when output size is small, a triangle join can be computed in time $\IN^{4/3}$~\cite{bpwtriangles}. This triangle computation algorithm uses matrix multiplication as a subroutine, and makes use of the fact that the matrix multiplication exponent $\omega$ is $< 3$ (The matrix multiplication exponent $\omega$ is defined as the smallest value such that two dense $N \times N$ matrices can be multiplied in time $O(N^{\omega})$). As another example, the graph with edges $XY_1$, $XY_2$, $Y_1Z_1$, $Y_2Z_1$, $XY_3$, $XY_4$, $Y_3Z_2$, $Y_4Z_2$, $Z_1Z_2$ can also be shown to have submodular width $2$. But we can compute its join in subquadratic time, again by using matrix multiplication in combination with the DARTS algorithm. 

\begin{theorem*}
Consider a graph with edges $XY_1$, $XY_2$, $Y_1Z_1$, $Y_2Z_1$, $XY_3$, $XY_4$, $Y_3Z_2$, $Y_4Z_2$, $Z_1Z_2$. A join over the graph can be solved in subquadratic time when output size is small.
\end{theorem*}
\begin{proof}(Sketch)
We briefly describe the transforms used to reduce the above join. First, if $X$ has degree $N^{\epsilon}$ for any $\epsilon > 0$, then a heavy transform reduces the join to an acyclic one, which means we can process the join in time $O(N^{2-\epsilon} + \OUT)$. So assume that $X$ has small degree. 

Then we perform a light transform on $\{X$, $Y_1$, $Y_2$, $Y_3$, $Y_4\}$, which gives a single relation of size $\approx N$ (since $X$ has low degree). Then we use a split transform to remove $X$, and we are left with edges $Y_1Y_2Y_3Y_4$, $Y_1Z_1$, $Y_2Z_1$, $Y_3Z_2$, $Y_4Z_2$, $Z_1Z_2$, all of size $N$. 

Now, if either $Z_1$ or $Z_2$ has degree $> N^{0.5 + \epsilon}$, we do a heavy transform on it, reducing the problem to a triangle join which can be solved in time $N^{3/2}$. In fact, if the degree of $Z_1$ is more than $d \times N^{\epsilon}$, while that of $Z_2$ is $d$ for any $d$, then we can do a heavy transform on $Z_1$, and the number of triangles for $Z_2$ is bounded by $Nd$, which gives us subquadratic time. So now we can assume that the degrees of $Z_1$ and $Z_2$ are almost equal, and less than $\sqrt{N}$. 

But if the degrees of $Z_1$ and $Z_2$ are less than $N^{0.25 - \epsilon}$ each, then a light transform on all attributes gives us an output with size $< N^{2-4\epsilon}$ (as each $Z_1,Z_2$ has at most $N^{1-4\epsilon}$ quadruples of neighbors.) So assume the degrees of $Z_1$, $Z_2$ are almost equal and between $N^{0.25}$ and $N^{0.5}$.

If the degrees  of $Z_1$, $Z_2$ are given by $d < N^{0.5-\epsilon}$, then we perform light transforms on $\{Z_1$,$Y_1$,$Y_2\}$ and $\{Z_2$,$Y_3$,$Y_4\}$, to get two triangles that have $< Nd$ tuples each. Then we perform a Split transform using articulation set $\{Z_1,Y_3,Y_4\}$. We can compute the join on attributes $Z_1$ and all the $Y$'s in time $N^{2-d}$ as there are $N^{1-d}$ $Z_1$ values and $N$ values of the $Y$'s. Thus the size bound on the projection onto $\{Z_1,Y_3,Y_4\}$ is also $N^2/d$. Then we can compute the join for $Z_1$, $Z_2$, $Y_3$, $Y_4$ in time $Nd^2$ since there are $Nd$ values of $Z_2Y_3Y_4$, and each $Z_2$ value has at most $d$ neighbors in $Z_1$. Thus we can solve this join in time $Nd^2 < N^{2-2\epsilon}$.

Now finally, assume that value in $Z_1$, $Z_2$ both have degree $d = N^{0.5}$. Like in the previous case, we perform a split transform on $Z_1$, $Y_3$, $Y_4$ and compute the join of $Z_1$ with all $Y$'s and their projection onto $Z_1Y_3Y_4$ in time $N^2/d = N^{3/2}$. But the other remaining join has relations $Z_1Y_3Y_4$, $Z_2Y_3Y_4$ and $Z_1Z_2$ of sizes $N^{3/2}$, $N^{3/2}$, $N$ respectively. We have $N^{1/2}$ values in $Z_1$, $Z_2$ and $N$ values in $Y_3Y_4$. We can convert $Y_3$, $Y_4$ into a single attribute with $N$ values to get a triangle join. Then we can randomly divide the $N$ values of $Y_3Y_4$ into $\sqrt{N}$ sets, to get $\sqrt{N}$ triangle joins (of three relations of size $N$ each). This is where we use matrix multiplication. Using the multiplication multiplication based algorithm for triangle finding~\cite{bpwtriangles}, we can solve each triangle join in time strictly less than $N^{3/2}$ when $\OUT$ is small. Then we can combine the solutions from the $\sqrt{N}$ triangle joins, and the total time taken is strictly less than $N^{3/2}$ $\times$ $\sqrt{N}$ $=$ $N^2$. The proves that the join can be solved in subquadratic time.
\end{proof}

\section{DBP Bound and Parallel Processing}

\subsection{Intuition behind the DBP bound}\label{sec:DBP-intuition}
The intuition behind the DBP bound is clearer when we use the dual version of Linear program~\ref{lp:degree-packing-primal}.
\begin{lp}\label{lp:degree-packing-dual}(Dual of Linear Program~\ref{lp:degree-packing-primal})
$$\textrm{Maximize} \sum_{(R,A) \in C, A' \subseteq A} w_{R,A'} \log\left(\frac{d_{\pi_{A}(R), A \setminus A'}}{L}\right)
\textrm{ s.t. } \forall a \in \mA : \sum_{(R,A) \in C, A' \subseteq A \mid a \in A'} w_{R,A'} \leq 1$$
\end{lp}

Linear program~\ref{lp:degree-packing-dual} is structurally similar to an edge packing program. In edge packing we assign a non-negative weight to each edge such that the total weight on each attribute is $\leq 1$, while maximizing the sum of all weights (weighted by log of the relation sizes). The linear program for $\DBP(\mR, 2)$ can be thought of as a variant of edge packing with the following differences:
\squishlist
\item Instead of assigning weights to only relations, we assign weights ($w_{R,A'}$) to subrelations $\pi_{A'}R$ as well.
\item We take a minimum over all covers of the join, where covers can consist of relations ($R$) or subrelations ($\pi_{A}(R)$). 
\item The biggest difference is, in edge packing the weight of each edge $\pi_{A'}(R)$ is multiplied by the log of its size. Here, instead of size, we use the maximum number of distinct values in $\pi_{A'}(R)$ that an external value (in $\pi_{A \setminus A'}(R)$) can connect to. This in-degree $d_{\pi_{A}(R), A \setminus A'}$ is naturally bounded by the size $|\pi_{A'}(R)|$ but can be smaller for sparse relations. 
\squishend

\subsection{Proof of Theorem~\ref{thm:gep-agm}}\label{proof:gep-agm}
\begin{theorem*}
For each degree configuration $c \in \mC_L$, the value of $\IN^{\DBP(\mR(c), L)}$ is $\leq$ to the AGM bound on $\mR(c)$. 
\end{theorem*}
\begin{proof}

For any relation $R \in \mR(c)$, and any $A \subseteq \attr(R)$, $d_{R,A}$ denotes the maximum degree of any value in $A$ in relation $R$. $d_{R, \emptyset}$ simply equals $|R|$. Note that the degree configuration $c$ specifies a degree bucket for each $(R,A)$. Let $d'_{R,A}$ denote the minimum degree of that bucket. The actual maximum degree $d_{R,A}$ may be strictly less than the values in bucket $Ld'_{R,A}$ because some of the neighbors of values in $A$ in the original relation may not be compatible with degree configuration $c$. The actual degree $d_{R,A}$ is also $\leq |\pi_{\attr(R) \setminus A}(R)|$. Now we define an {\em effective size} $S(R,A)$ for any pair $(R,A)$ inductively: 
\squishlist
\item $S(R,\emptyset)  = 1$
\item $S(R,A) = \max_{A' \subsetneq A} S(R,A') \times \frac{d_{\pi_{A}(R), A'}}{L}$
\squishend

If $A \neq \emptyset$, then setting $A' = \emptyset$ in the definition tells us that $S(R,A) \geq S(R,\emptyset) \times \frac{d_{\pi_{A}(R), \emptyset}}{L} = \frac{|\pi_{A}(R)|}{L}$. This tells us that $S(R,A)$ is lower bounded by the actual size of $\pi_{A}(R)$ divided by $L$. 
We can inductively prove an upper bound on $S(R,A)$, by its maximum possible size divided by $L$. Specifically, for $A \neq \emptyset$:
$$S(R,A) \leq \frac{|R|}{d'_{R,A}L}$$

This is easily true for singleton $A$s, since their $S$ is simply equal to $\frac{|\pi_{A}(R)|}{L} \leq \frac{|R|}{d'_{R,A}L}$.
For bigger $A$s, we can prove this as follows: Each $A'$ value in the current configuration has at most $Ld'_{R,A'}$ neighbors in the original $R$. Each $A$ value in the current configuration has at least $d'_{R,A}$ neighbours in the original $R$. Thus, each $A'$ value in the current configuration has at most $\frac{Ld'_{R,A'}}{d'_{R,A}}$ neighbors in $\pi_{A}(R)$ in the current configuration i.e. $d_{\pi_{A}(R), A'} \leq \frac{Ld'_{R,A'}}{d_{R,A}}$. Now in the definition of $S(R,A)$, if $A' = \emptyset$, then we again get 
\begin{align*}
S(R,\emptyset) \times \frac{d_{\pi_{A}(R), \emptyset}}{L} \leq 1 \times \frac{|\pi_{A}(R)|}{L} \leq \frac{|R|}{d'_{R,A}L}
\end{align*}
For $A' \neq \emptyset$, we have
\begin{align*}
S(R,A') \times \frac{d_{\pi_{A}(R), A'}}{L} \leq \frac{|R|}{d'_{R,A'}L} \times \frac{Ld'_{R,A'}}{d_{R,A}L} = \frac{|R|}{d_{R,A}}
\end{align*}

We prove the result by giving a sequence of linear programs, starting from the dual of the fractional cover program (whose optimal objective value equals the log of the AGM bound), and ending with the DBP program (whose optimal objective value equals log of the DBP bound), such that the optimal objective value in each step is less than or equal to that in the previous step.
\begin{enumerate}
\item To start with, we have the dual of the fractional cover linear program, that assigns a non-negative value $v_a$ to each attribute $a$ such that for each relations $R$ in the join, the sum of values of attributes assigned to that relation is less than log of the relation size $|R|$. The objective is to maximize the sum of the $v_a$s. The optimal objective value for this program gives us the AGM bound. 
\item We modify the program to include constraints for subrelations. That is, for each $R$, for each $A \subseteq \attr(R)$, we add a constraint saying that the some of values of attributes in $A$ must be $\leq \log\left(\frac{|R|}{d_{R,A}}\right)$. The program is still feasible (since all $v_a$s equal to zero is a valid solution), but more constrained than the previous one. Since it is a maximization problem, additional constraints can only reduce the optimal objective value.
\item We reduce the right hand sides of the constraints from $\frac{|R|}{d_{R,A}}$ to $S(R,A)$. Since $S(R,A) \leq \frac{|R|}{d_{R,A}}$ for each $R,A$, the resulting program is strictly more constrained, while still being feasible, and hence its optimal objective value is less than or equal to the previous program.  
\item Now we actually consider an optimal solution to the linear program. Some of the constraints must be tight in the optimal solution. Moreover for each attribute $a$, there must exist a tight constraint $(R,A)$ such that $a \in A$, because otherwise we could increase $v_a$ slightly, increasing the objective value, without violating any constaints, which contradicts the optimality of our solution. That is, the set of tight constraints $(R,A)$ form a cover of the attributes. Call the cover $C$. Replace the inequality constraints for $(R,A) \in C$ with equality constraints. The resulting program is more constrained, but the previous optimal solution is feasible for this program as well, so it has the exact same optimal objective value. 
\item Now for each $(R,A) \in C$ and each $A' \subseteq A$, we have an equality constraint $\sum_{a \in A} v_a = \log(S(R,A))$ and and inequality constraint $\sum_{a \in A'} v_a \leq \log(S(R,A'))$. Together, these constraints imply $\sum_{a \in A \setminus A'} v_a \geq \log\left(\frac{S(R,A)}{S(R,A')}\right)$. Thus, for each $(R,A) \in C, A' \subsetneq A$, we keep the equality constraint $\sum_{a \in A} v_a = \log(S(R,A))$, but replace $\sum_{a \in A'} v_a \leq \log(S(R,A'))$ with $\sum_{a \in A\setminus A'} v_a \geq \log\left(\frac{S(R,A)}{S(R,A')}\right)$. This gives an equivalent linear program, which hence has the same optimal objective as before. Note that by replacing $A'$ with $A \setminus A'$, we can rewrite the above constraint as $\sum_{a \in A'} v_a \geq \log\left(\frac{S(R,A)}{S(R,A \setminus A')}\right)$.
\item Now, we keep constraints the same, but try to minimize rather than maximize the objective. The resulting program is still feasible, but may have a smaller objective value. The value won't be zero because now we have $\geq \log\left(\frac{S(R,A)}{S(R,A \setminus A')}\right)$ constraints for the $R,A,A'$s.
\item Earlier, we had only changed constraints for $R,A,A'$ where $(R,A)$ belonged to cover $C$ and $A'$ was a subset of $A$ (turning then from $\leq$ constraints to $\geq$ constraints). Thus, from our original dual program, we may have leftover $\leq$ constraints for $A'$ that are not the subset of any $A$ in the cover. We drop these constraints. The resulting problem is now less constrained than earlier, and since it is a minimization problem, the resulting objective can only be smaller. 
\item For $A' \subsetneq A$, the inductive definition of $S$ tells us that $\frac{S(R,A)}{S(R,A \setminus A')} \geq \frac{d_{\pi_{A}(R),A \setminus A'}}{L}$. We change the RHS of the $R,A,A'$ constraints from $\log\left(\frac{S(R,A)}{S(R,A \setminus A')}\right)$ to $\log\left(\frac{d_{\pi_{A}(R),A \setminus A'}}{L}\right)$. This only loosens the constraints. For each $(R,A) \in C$, we currently have an equality constraint $\sum_{a \in A} v_a = \log(S(R,A))$. We use the known lower bound on $S(R,A)$ to replace the equality constraint by $\sum_{a \in A} v_a \geq \log\left(\frac{|\pi_{A}(R)|}{L}\right)$. This also loosens the constraints. Note that since $d_{\pi_{A}(R),\emptyset} = |\pi_{A}(R)|$, this constraint is actually now a special case of the constraints with $R, A, A'$. Since both the above steps loosen the constraints, this can only decrease the optimal objective value.  
\item The resulting linear program can be seen to be the program used to define DBP, with an extra $\frac{1}{L}$ factor in the RHS of each constraint. As $L$ becomes smaller, the optimal objective value of the program tends to that of the DBP program. Moreover, since DBP itself is a minimum over all covers, while for this program we chose a specific cover, the actual DBP is less than the solution to this linear program, which is less than the AGM bound.
\end{enumerate} 

This proves the result, as required. 
\end{proof}

If $L$ is less than the size of each relation, and $\rho*$ is the fractional cover of the join query (used in the AGM bound), then in fact $\DBP(\mR(c), L) \leq L^{-\rho*}\AGM$. This can be seen by replacing the right hand sides of the constraints of the program in step $1$ by $\frac{|R|}{L}$ instead of $|R|$. This reduces the objective value of the original program, and the remaining steps still go through. 

\subsection{Examples comparing the DBP and AGM bounds}\label{sec:dbp-agm-examples}
\begin{example} ({\bf Comparison between DBP and AGM})

Let $L = 2$ for this example. Consider a triangle join $R(X,Y) \Join S(Y,Z) \Join T(Z,X)$. Let $|R| = |S| = |T| = N$. Let the degree of each value $x$ in $X$, in $R$ and $T$ be $d$. For different values of $d$, we will choose a cover $C$ and find the objective value of the linear program for that cover. Note that the DBP bound is a minimum over all covers, so it is possible that a different cover $C^*$ gives an even smaller linear program objective, but the purpose of this example is to show that the DBP bound can be much tighter than the AGM bound; hence it suffices to show that an `upper bound' on the DBP bound is much tighter than the AGM bound.

{\bf Case 1. $d < \sqrt{N}$:} We choose cover $C = \left\lbrace (R, \left\lbrace X,Y \right\rbrace), (T, \left\lbrace X,Z \right\rbrace) \right\rbrace$. For this cover, the solution to Linear Program~\ref{lp:degree-packing-dual} is $w_{R, \left\lbrace X,Y \right\rbrace} = w_{T, \left\lbrace Z \right\rbrace} = 1$ with all other values set to $0$. The objective value is $\log(N) + \log(d) = \log(Nd)$. Thus, the DBP bound is $\leq Nd$, which tells us that join output size is upper bounded by $Nd$. 

{\bf Case 2. $d > \sqrt{N}$:} Since $d$ is large, the number of distinct $X$ values must be small. To take advantage of this, we consider cover $C' = \left\lbrace (R, \left\lbrace X \right\rbrace), (S, \left\lbrace Y,Z \right\rbrace) \right\rbrace$. Now the linear program solution is trivially $w_{R, \left\lbrace X \right\rbrace} = w_{S, \left\lbrace Y,Z \right\rbrace} = 1$, which gives us the join size bound of $\frac{N^2}{d}$ (since $d_{R, \left\lbrace X \right\rbrace} \leq |\pi_{X}(R)| \leq \frac{N}{d}$). 

In contrast, the AGM bound gives us a loose upper bound of $N^{\frac{3}{2}}$ irrespective of degree $d$. Computing the AGM bound individually over each degree configuration does not help us do better, as the above example can have all tuples in a single degree configuration.
\end{example}

\begin{example}
As suggested by the above example, the DBP bound has a tighter exponent than the AGM bound for almost all possible degrees (namely, degrees higher or lower than $\sqrt{N}$). As a more general example, suppose we have a join consisting of binary relations of size $N$ each, where each value has degree $d$, where the join hypergraph is connected. Then the AGM bound on this join will equal the DBP bound only when $d \approx \sqrt{N}$. If $d < \sqrt{N}^{1 - \epsilon}$, then the DBP bound will be smaller than the AGM bound by a factor of at least $N^{\frac{\epsilon}{2}}$. 

To show this, consider a traversal of the join hypergraph $X_1, X_2,\ldots,X_n$ such that $R(X_1,X_2)$ is a relation in the join, and for all $i > 2$, there is a $j < i$ such that $X_j, X_i$ is a relation (call it $R(i)$) in the join. Then consider cover $C = \{(R, \{X_1,X_2\})\} \cup \{(R(i), \{X_i\}) \mid i > 2\}$. The solution to the linear program is $w_{R',A'} = 1$ for all $(R',A') \in C$ and $0$ otherwise. This gives us a bound of $N \times d^{|C|-1} = Nd^{n-2}$. In contrast, if we have $n$ attributes, the AGM bound must be at least $\sqrt{N}^n$ (which is actually achieved if all attributes have $\sqrt{N}$ values and all relations are full cartesian products). Thus the ratio of the AGM bound to the DBP bound is at least $(\frac{\sqrt{N}}{d})^{n-2} > \sqrt{N}^{\epsilon} = N^{\frac{\epsilon}{2}}$.

On the other hand, $d$ cannot be $> \sqrt{N}^{1 + \epsilon}$ for all values, because if it is (say in relation $R(X,Y)$, then the number of values in attribute $X$ must be $O(\sqrt{N}^{1-\epsilon})$ which is this smaller than the degree of values in $Y$.
\end{example}

\subsection{Proof of Lemma~\ref{lemma:shares-communication-load}}
\begin{lemma*}
The shares algorithm, where each attribute $a$ has share $\IN^{v_a}$, where $v_a$ is from the solution to Linear Program~\ref{lp:degree-packing-primal}, has a load of $O(L)$ per processor with high probability, and a communication cost of $O(\max_{c \in \mC_L}L \cdot \IN^{\DBP(\mR(c), L)})$.
\end{lemma*}
\stitle{Communication:} Consider any $(R, A) \in C$. As per the shares algorithm, every tuple in $\pi_{A}(R)$ will have to be sent to every processor whose hash value in $A$ matches that of the tuple. Thus, the number of processors to which each tuple is sent is given by $\Pi_{a \notin A} \IN^{v_a}$. Thus, total communication for $R, A$ is given by 
$$|\pi_{A}(R)| \times \Pi_{a \notin A} \IN^{v_a} \leq L \cdot \IN^{\sum_{a \in A} v_a} \times \IN^{\sum_{a \notin A} v_a} = L \cdot \IN^{\DBP(\mR(c), L)}$$

Thus, total communication is bounded by $L \cdot \IN^{\DBP(\mR(c), L)}$ (multiplied by some factors that depend on the number of relations and schema sizes, but not on the number of tuples in the relations). 

\stitle{Load:} Now we analyze load per processor. We will show that the $m^{th}$ moment of load on a processor is $O(L^{m})$, which shows that the load is $O(L)$ with high probability, ignoring factors not depending on $\IN$. Consider an $(R, A) \in C$, and a processor with hash value $h_1$ for $A$ and $h_2$ for remaining attributes. Each tuple of $\pi_{A}(R)$ will be sent to this processor if its hash on $A$ equals $h_1$. For any value $x \in \pi_{A}(R)$, let $I_{x}$ be an indicator variable thats true if the hash of $x$ equals $h_1$. Then expected load on the processor from $(R,A)$ is
$$ E\left[\text{Load}\right] = \sum_{x \in \pi_{A}(R)} E[I_{x}] \leq L \cdot \IN^{\sum_{a \in A} v_a} \times \IN^{\sum_{a \in A} -v_a} = L$$

Now let us consider the $m^{th}$ moment of the load. Consider $m$ tuples $t_1, t_2,\ldots, t_m \in \pi_{A}(R)$. Each tuple specifies a value in each attribute in $A$. Some of these values may be equal to each other. For example, for tuples $(x, y)$ and $(x, y')$, the first value is equal. We are going to count the number of $m$-sized sets of tuples with the same pattern of equal values, and the probability of all these tuples being sent to the processor and show that it is $O(L^{m}$. Define $T_l$ for $1 \leq l \leq m$ to be the set of attributes whose values in $t_l$ occur in $t_l$ but not in $t_1, t_2,\ldots, t_{l-1}$. For instance, if $R$ had schema $(X, Y, Z)$ and we had tuples $t_1 = (x_1, y_1, z_1), t_2 = (x_1, y_2, z_2), t_3 = (x_2, y_2, z_1), t_4 = (x_2, y_3, z_3)$, then we would have $T_1 = \left\lbrace X, Y, Z\right\rbrace, T_2 = \left\lbrace Y, Z\right\rbrace, T_3 = \left\lbrace X \right\rbrace, T_4 = \left\lbrace Y, Z\right\rbrace$. $T_1$ is always equal to $A$ by this definition. The probability of all these tuples being hashed to a given processor is $\Pi_{l=1}^{m}(\Pi_{a \in T_l} \IN^{-v_a})$. The number of such tuple sets is upper bounded by $\Pi_{l=1}^{m}d_{\pi_{A}(R),A \setminus T_l}$, since the number of ways of choosing $t_l$ such that its $A \setminus T_l$ part is fixed, is $d_{\pi_{A}(R),A \setminus T_l}$. Thus, the $m^{th}$ moment of the load is:
$$\Pi_{l=1}^{m}d_{\pi_{A}(R),A \setminus T_l} \times \Pi_{l=1}^{m}(\Pi_{a \in T_l} \IN^{-v_a}) \leq L^m\Pi_{l=1}^{m}\IN^{\sum_{a \in T_l} v_a} \times \Pi_{l=1}^{m}(\Pi_{a \in T_l} \IN^{-v_a}) = L^m $$
Thus, the $m^{th}$ moment of load is $O(L^m)$, and so the load per processor is $O(L)$ with high probability, ignoring terms not depending on $\IN$. 

\subsection{Additional Examples for the parallel algorithm}\label{sec:parallel-algo-examples}

\begin{example}
Generalizing the previous example, let the degree of each value be $O(\delta)$, where $\delta < \sqrt{N}$. Let $p$ be the required number of processors at load level $L$.
\squishlist
\item If $L < \delta$, then $p = \DBP(\mR,L) = \frac{N\delta}{L^2}$. 
\item If $\delta \leq L < \frac{N}{\delta}$, then $p = \DBP(\mR,L) = \frac{N}{L}$. 
\item If $\frac{N}{\delta} \leq L < N$, then $p = \DBP(\mR, L) = 1$.
\squishend

Now we invert the above analysis to see how changing the number of processors $p$ changes load $L$. When $p = 1$ we have $L = N\delta^{-1}$. As $p$ increases up to $N\delta^{-1}$, the load is $Np^{-1}$. So as long as $p \leq N\delta^{-1}$, we get optimal parallelism. Beyond that, as $p$ increases to $N\delta$, load decreases as $\sqrt{N\delta p^{-1}}$. Thus, beyond $N\delta^{-1}$, doubling $p$ gives us only a $\sqrt{2}$ reduction in load. Finally, when $p = N\delta$, the load becomes $O(1)$, which is the maximum parallelism level.
\end{example}

\begin{example}
In this example, we demonstrate that our parallel algorithm can even outperform existing optimal sequential algorithms : 

Consider a triangle join $\mR = \left\lbrace R_1(X,Y), R_2(Y,Z), R_3(Z,X) \right\rbrace$. Let $|R_1| = |R_2| = |R_3| = N$. Also suppose $|Z| = N$, $|X| = |Y| = \sqrt{N}$, and the degrees of all $z \in Z$ are $O(1)$ while degrees of values in $X, Y$ are $O(\sqrt{N})$. The DBP bound on the join is $O(N)$. Running worst case optimal algorithms like NPRR and LFTJ take time $N^{\frac{3}{2}}$ to process the join if the attribute order is $X, Y, Z$ or $Y, X, Z$. On the other hand, a simple sequentialized version of our parallel algorithm takes time $O(N)$. By concatenating three such joins, with a different attribute being the {\em sparse} attribute each time, we get a join for which NPRR takes time $N^{\frac{3}{2}}$ for all attribute orders, while our sequentialized parallel algorithm takes time $O(N)$.

Note that using GHD based algorithms (that have runtime $O(\IN^{fhw}+\OUT)$) does not improve the $O(N^{\frac{3}{2}})$ runtime, as all three relations must be in a single bag of the GHD.
\end{example}

\end{document}